\documentclass[american,aps,pra,reprint,floatfix,nofootinbib,superscriptaddress,longbibliography]{revtex4-1}
\usepackage[unicode=true,pdfusetitle, bookmarks=true,bookmarksnumbered=false,bookmarksopen=false, breaklinks=false,pdfborder={0 0 0},backref=false,colorlinks=false]{hyperref}
\hypersetup{colorlinks,linkcolor=myurlcolor,citecolor=myurlcolor,urlcolor=myurlcolor}
\usepackage{graphics,epstopdf,graphicx,amsthm,amsmath,amssymb,mathptmx,braket,colortbl,color,bm,framed,mathrsfs}
\usepackage[T1]{fontenc}
\usepackage[up]{subfigure}
\usepackage{tikz}

\definecolor{myurlcolor}{rgb}{0,0,0.9}

\newcommand{\proj}[1]{| #1\rangle\!\langle #1 |}

\DeclareMathOperator{\trace}{Tr}

\newcommand{\Ptr}[2]{\trace_{#1}\Pa{#2}}
\newcommand{\Tr}[1]{\Ptr{}{#1}}

\newcommand{\Pa}[1]{\left[#1\right]}
\newcommand{\Br}[1]{\left[#1\right]}

\newcommand{\norm}[1]{\left\lVert #1 \right\rVert}

\theoremstyle{plain}
\newtheorem{thm}{Theorem}
\newtheorem{lem}[thm]{Lemma}
\newtheorem{prop}[thm]{Proposition}
\newtheorem{cor}[thm]{Corollary}
\theoremstyle{definition}
\newtheorem{remark}[thm]{Remark}

\newcommand*{\myproofname}{Proof}

\def\ot{\otimes}
\def\complex{\mathbb{C}}
\def\real{\mathbb{R}}

\def \diag {\mathrm{diag}}

\DeclareMathAlphabet{\mathcal}{OMS}{cmsy}{m}{n}

\makeatother

\begin{document}

  \author{Kaifeng Bu}
 \email{kfbu@fas.harvard.edu}
\affiliation{Department of Physics, Harvard University, Cambridge, Massachusetts 02138, USA}

\author{Dax Enshan Koh}
  \email{dax\_koh@ihpc.a-star.edu.sg}
\affiliation{Institute of High Performance Computing, Agency for Science, Technology and Research (A*STAR), 1 Fusionopolis Way, \#16-16 Connexis, Singapore 138632, Singapore}

\author{Lu Li}
\affiliation{Department of Mathematics, Zhejiang Sci-Tech University, Hangzhou, Zhejiang 310018, China}
\affiliation{School of Mathematical Sciences, Zhejiang University, Hangzhou, Zhejiang 310027, China}

\author{Qingxian Luo}
\affiliation{School of Mathematical Sciences, Zhejiang University, Hangzhou, Zhejiang 310027, China}
\affiliation{Center for Data Science, Zhejiang University, Hangzhou, Zhejiang 310027, China}

\author{Yaobo Zhang}
\affiliation{Zhejiang Institute of Modern Physics, Zhejiang University, Hangzhou, Zhejiang 310027, China}
\affiliation{Department of Physics, Zhejiang University, Hangzhou, Zhejiang 310027, China}

\title{On the statistical complexity of quantum circuits}

\begin{abstract}

In theoretical machine learning, the statistical complexity is a notion that measures the richness of a hypothesis space.
In this work, we apply a particular measure of statistical complexity, namely the Rademacher complexity, to the quantum circuit model in quantum computation and study how the statistical complexity depends on various quantum circuit parameters.
In particular, we investigate the dependence of the statistical complexity on the resources,  depth, width, and the number of input and output registers of a quantum circuit.
To study how the statistical complexity scales with resources in the circuit, we introduce a resource measure of magic based on the $(p,q)$ group norm, which quantifies the amount of magic in the quantum channels associated with the circuit.
These dependencies are investigated in the following two settings:
(i) where the entire quantum circuit is treated as a single quantum channel, and (ii) where each layer of the quantum circuit is treated as a separate quantum channel. 
The bounds we obtain can be used to constrain the capacity of
quantum neural networks in terms of their depths and widths as well as the resources in the network.

\end{abstract}

\maketitle

\section{Introduction}

Owing to its ability to recognize and analyze patterns in data and use them to make predictions, deep learning---a subfield of machine learning---has made a profound impact on the computing industry \cite{lecun2015deep, goodfellow2016deep,pml1Book} and has found applications in a myriad of fields, including natural language processing \cite{young2018recent, deng2018deep, li2017deep}, drug design \cite{jing2018deep,gawehn2016deep}, fraud detection \cite{roy2018deep,pumsirirat2018credit}, medical image analysis \cite{shen2017deep, litjens2017survey}, self-driving cars \cite{ramos2017detecting,rao2018deep}, handwriting recognition \cite{pham2014dropout,ghosh2017comparative}, and computer vision \cite{voulodimos2018deep, ponti2017everything}. A central object in many deep learning models is the neural network, an interconnected collection of nodes that can learn from data and model relationships between them \cite{nielsen2015neural}. Different neural networks differ in terms of their ability to learn from data, and understanding this difference is a key problem in theoretical machine learning. This ability of neural networks has been quantified by various statistical complexity measures, including the Vapnik–Chervonenkis (VC) dimension \cite{Vapnik71,Vapnik82}, the metric entropy \cite{tikhomirov1993varepsilon}, the Gaussian complexity \cite{Bartlett03}, and the Rademacher complexity \cite{Bartlett03}.
The dependence of these measures on various structure parameters of the neural network, such as its depth and width and the number 
of parameters in the neural network, has been studied in a number of papers \cite{Telgarsky16,Neyshabur15,Harvery17, Bartlett17,Neyshabur2017,Golowich18a}. 

In addition to the progress in deep learning, the last decade also saw rapid developments in quantum computing \cite{national2019quantum}. With the development of noisy intermediate-scale quantum (NISQ) hardware \cite{preskill2018quantum} as well as near-term quantum algorithms like the variational quantum eigensolver (VQE) \cite{peruzzo2014variational} and the quantum approximate optimization algorithm (QAOA) \cite{farhi2014quantum, zhou2020quantum}, there are expectations that 
quantum computers are poised to revolutionize computation by speeding up the solutions of certain practical computational problems \cite{cerezo2020variational}. Major experimental milestones in this direction include the recent demonstrations of quantum computational supremacy \cite{arute2019quantum, zhong2020quantum} (also called quantum advantage \cite{Palacios-Berraquero2019Dec}), defined to be an event in which a quantum computer empirically solves a computational problem deemed intractable for classical computers, independent of the practical value of the problem \cite{preskill2012quantum, lund2017quantum, harrow2017quantum, dalzell2020many}.

At the intersection of deep learning and quantum computing is the field of quantum deep learning, which has the quantum neural network---the quantum generalization of the classical neural network---as one of its central objects \cite{farhi2018classification,beer2020training, Sharma2020,Schuld2014Nov,killoran2019continuous,Cong2019Dec}. Quantum deep learning has been explored as an application of quantum machine learning, which has gained significant interest of late \cite{lloyd2013quantum,wittek2014quantum, BiamonteNature17,ciliberto2018quantum,Dunjko2018}. Compared to the classical neural networks, however, considerably less is known about quantum neural networks and characterizations of their statistical complexities. For example, the following question has hitherto remained largely unaddressed: how does the statistical complexity of quantum neural networks depend on
the structure parameters of the quantum circuit underlying it as well as the amount of certain resources it contains?

In this paper, we address the above gap by characterizing the statistical complexity of quantum circuits in terms of their Rademacher complexity. To characterize the dependence of Rademacher complexity on resources in the framework of quantum resource theories \cite{COECKE16,chitambar_2019}, 
we introduce a resource measure of magic \cite{howard_2017} for quantum channels based on the $(p,q)$ group norm. 
We consider the Rademacher complexity of quantum circuits in two different settings. First, we consider the case where the entire quantum circuit is treated as a single quantum channel independent of its depth or width.
In this case, we find a bound for the statistical complexity that depends on a resource measure of magic as well as the number of input and output qubits. Second, we consider the case where each layer of the quantum circuit is treated as a separate quantum channel. In this case, we find a bound for the statistical complexity that depends not only on the resource measure of magic but also on the depth and width of the quantum circuit.

\section{Main results}

Consider $m$ independent samples $S=(\vec{x_1},\ldots,\vec{x}_m)$, where each $\vec{x}_i$ is encoded as a quantum state
$\ket{\psi(\vec{x}_i)}$. After a quantum circuit $C$ (e.g., $C$ could be an instance of a variational quantum circuit or a quantum neural network) is applied to the quantum state $\ket{\psi(\vec{x}_i)}$ and a (Hermitian) observable $H$ is measured on the output, 
the expected measurement outcome 
is given by 
\begin{eqnarray}\label{eq:fun_gen}
f_{C}(\vec{x}_i)
=\Tr{C(\proj{\psi(\vec{x}_i)})H}.
\end{eqnarray}
In this way, each quantum circuit $C$ defines a real-valued function $f_{C}$.
Let $\mathcal{F}\circ\mathcal{C}: = \{f_C : C \in \mathcal{C}\}$ denote the function class defined 
by  the set of quantum circuits $\mathcal{C}$.

Consider the hypothesis space $\mathcal{H}=\mathcal{F}\circ\mathcal{C}$,
where $\mathcal{C}$ is a given set of quantum circuits. Given $m$ independent samples $\set{(\vec{x}_i,y_i)}^{m}_{i=1}$, where each
$(\vec{x}_i, y_i)$ is taken i.i.d. from some unknown probability distribution $D$ on some $\mathcal X\times\mathcal Y$, 
let us consider a loss function $l:\mathcal{Y}\times \mathcal{Y}\to \real$. The goal of the learning task is to find some 
function in the hypothesis space that minimizes the expected error $L(f)=\mathbb{E}_{(\vec{x},y)\sim D}l(f(\vec{x}),y)$.
As we have access to only the $m$ independent samples $\set{(\vec{x}_i,y_i)}^{m}_{i=1}$, one strategy is to find some 
function in hypothesis space to minimize the empirical  error $\hat{L}(f)=\frac{1}{m}\sum^m_{i=1}l(f(\vec{x}_i),y_i)$.
The difference between the empirical and expected error is called the \textit{generalization error}, which determines the performance 
of the hypothesis function $f$ on the unseen data drawn from the unknown probability distribution.

The Rademacher complexity is a measure of the richness of a hypothesis space 
and can be used to provide bounds on the generalization error associated with learning from training data \cite{Bartlett03,koltchinskii2006}.
Let us consider the Rademacher complexity of $\mathcal{F}\circ\mathcal{C}$ 
 on $m$ independent samples $S=\set{\vec{x_1},...,\vec{x}_m}$, defined as
 \begin{eqnarray}
 R_S(\mathcal{F}\circ\mathcal{C})
 =\mathbb{E}_{\vec{\epsilon}}
 \frac{1}{m}\sup_{C\in\mathcal{C} }
\left|\sum_i\epsilon_i f_{C}(\vec{x}_i)\right|,
 \end{eqnarray}
where each $\epsilon_i$ in the expectation above is a Rademacher random variable, which takes the values $\pm 1$ with equal probability $1/2$.
Here, we use the Rademacher complexity as a measure of the statistical complexity of the hypothesis space $\mathcal{F}\circ\mathcal{C}$.

\subsection{Rademacher complexity of  quantum channels}

\subsubsection{Rademacher complexity of arbitrary quantum channel}

Given a quantum channel $\Phi:\mathcal L((\complex^2)^{\ot n_1})\to \mathcal L((\complex^2)^{\ot n_2})$ from $n_1$ qubits to $n_2$ qubits, we define the $4^{n_2}\times 4^{n_1}$ \textit{representation matrix}  $M^{\Phi}$ of $\Phi$ to be the matrix whose entries are given by 
\begin{eqnarray}
M^{\Phi}_{\vec{z}\vec{x}}
=\frac{1}{2^{n_2}}\Tr{P_{\vec{z}}\Phi(P_{\vec{x}})},
\end{eqnarray} 
where $\vec{x}\in \set{0,1,2,3}^{n_1}$, $\vec{z}\in \set{0,1,2,3}^{n_2}$, and $P_{\vec{x}}$, $P_{\vec{z}}$ are the corresponding Pauli operators. 
For any Hermitian  operator $P$, the \textit{representation vector} $\vec{\alpha}^P$ of $P$ is defined 
as 
\begin{eqnarray}
\alpha^{P}_{\vec{z}}
=\frac{1}{2^{n}}\Tr{P_{\vec{z}}P}.
\end{eqnarray} 

For any $N_1\times N_2$ matrix $M$, 
which can be treated as a column of $N_1$ row vectors,
the $(p,q)$ \textit{group norm} of $M$, where $0<p,q\leq \infty$, is defined as
$
\norm{M}_{p,q}=
\left(\frac{1}{N_1}\sum_{i}\norm{M_i}^q_p\right)^{1/q}
$, 
where the $l_p$ norm of the $i$-th row vector $\norm{M_i}_p$ is defined as 
$
\norm{M_i}_p
=\left(\sum^{N_2}_{j=1}
|M_{ij}|^p\right)^{1/p}
$.
Of interest to us is the $(p,q)$ group norm of the representation matrix of quantum channels. As we shall show in Appendix \ref{apen:single}, the $(p,q)$ group norm of the representation matrix of quantum gates can be 
used as a resource measure to quantify the amount of magic in the quantum gates.

Here, we treat the entire quantum circuit as a single quantum channel.
Let us define $\mathcal{C}^{n_0,n_1}_{\norm{\cdot}_{p,q}\leq \mu}$ to be the set of quantum circuits from $n_0$ qubits to $n_1$ qubits that have a $(p,q)$ group norm bounded by $\mu$.

\begin{thm}
\label{thm:one}
Given the set of  quantum circuits $C$ from $n_0$ qubits to $n_1$ qubits with bounded $(p,q)$ norm $\norm{\cdot}_{p,q}$,
 the Rademacher complexity of $\mathcal{F}\circ \mathcal{C}^{n_0,n_1}_{\norm{\cdot}_{p,q}\leq \mu}$ 
 on $m$ independent samples $S=\set{\vec{x_1},...,\vec{x}_m}$ is bounded as follows: 
 
 (1) For $1\leq p\leq 2$, we have
\begin{eqnarray}
R_S(\mathcal{F}\circ \mathcal{C}^{n_0,n_1}_{\norm{\cdot}_{p,q}\leq \mu})
\leq \mu 4^{n_1\max\set{\frac{1}{p^*},\frac{1}{q}}}
\frac{\sqrt{\min\set{p^*, 8n_0}}}{\sqrt{m}}K_p(S,H).
\nonumber\\
\end{eqnarray}

(2) For $2<p<\infty$, we have
\begin{eqnarray}
R_S(\mathcal{F}\circ\mathcal{C}^{n_0,n_1}_{\norm{\cdot}_{p,q}\leq \mu})
\leq \mu 4^{n_1\max\set{\frac{1}{p^*},\frac{1}{q}}}
\frac{\sqrt{p^*}}{m^{1/p}}K_p(S,H),
\end{eqnarray}
where $p^*$ is the H\"{o}lder conjugate of $p$, i.e., $\frac{1}{p}+\frac{1}{p^*}=1$; 
\begin{equation}\label{eq:KpSH}
    K_p(S,H)=\norm{\vec{\alpha}}_p\max_i\norm{\vec{f}_I(\vec{x}_i)}_{p^*};
\end{equation} 
and $\vec{\alpha}$ and $\vec{f}_I(\vec{x}_i)$  are the representation vectors of $H$ and $\proj{\psi(x_i)}$ in the Pauli basis, respectively.

\end{thm}
This result provides an upper bound on the Rademacher complexity of quantum circuits
that depends on the amount of magic and the number of input and output qubits. (See  Appendix \ref{apen:single} for a proof of Theorem \ref{thm:one}.)

\subsubsection{Rademacher complexity of unital quantum channels}

We now consider the special case where the quantum channel $\Phi$ is unital, i.e., $\Phi(\mathbb{I})=\mathbb{I}$. In this case, the representation matrix $M^{\Phi}$ has the following form
$
M^{\Phi}=\Br{
\begin{array}{ccc}
1&\vec{0}^T\\
\vec{0}& \hat{M}^{\Phi}
\end{array}
}
$.
We shall define the \textit{modified representation matrix} $\hat{M}^\Phi$ to be the bottom-right $(4^{n_2}-1)\times (4^{n_2}-1)$
submatrix of $M^{\Phi}$. 
Next, note that the representation vector of a Hermitian operator $P$ can be written as $\vec{\alpha}^P=(\alpha_0,\hat{\vec{\alpha}}^P)$. We shall call $\hat{\vec{\alpha}}^P$ the \textit{modified representation vector} of the operator $P$.

For a unital channel $\Phi$, we shall denote the $(p,q)$ group norm of the modified
 representation matrix $\hat{M}^{\Phi}$ as $\norm{\hat{M}^\Phi}_{p,q}$.
Note that the $(p,q)$ group norm of the modified representation matrix of unital quantum channels can be regarded as a resource measure of magic  (see Appendix \ref{apen:single_unital}).

Similarly, let us define  $\mathcal{C}^{n_0,n_1}_{\norm{\hat{\cdot}}_{p,q}\leq \mu}$ to be the set of unital quantum circuits $C$ from $n_0$ qubits to $n_1$ qubits with bounded norm $\norm{\hat{\cdot}}_{p,q}$.

\begin{thm}\label{thm:two}
Let $H$ be a traceless observable.
Given a set of unital  quantum circuits from $n_0$ qubits to $n_1$ qubits with bounded norm $\norm{\hat{\cdot}}_{p,q}$, the Rademacher complexity 
of $\mathcal{F}\circ\mathcal{C}^{n_0,n_1}_{\norm{\hat{\cdot}}_{p,q}\leq \mu}$
 on $m$ samples $S=\set{\vec{x_1},...,\vec{x}_m}$ is bounded as follows.
 
 (1) For $1\leq p\leq 2$, we have
\begin{eqnarray}
R_S(\mathcal{F}\circ\mathcal{C}^{n_0,n_1}_{\norm{\hat{\cdot}}_{p,q}\leq \mu})
\leq \mu N^{\max\set{\frac{1}{p^*},\frac{1}{q}}}_1
\frac{\sqrt{\min\set{p^*, 8n_0}}}{\sqrt{m}}\hat{K}_p(S,H).
\nonumber\\
\end{eqnarray}

(2) For $2<p<\infty$, we have
\begin{eqnarray}
R_S(\mathcal{F}\circ\mathcal{C}^{n_0,n_1}_{\norm{\hat{\cdot}}_{p,q}\leq \mu})
\leq \mu N^{\max\set{\frac{1}{p^*},\frac{1}{q}}}_1
\frac{\sqrt{p^*}}{m^{1/p}}\hat{K}_p(S,H),
\end{eqnarray}

where $N_1=4^{n_1}-1$,
\begin{equation}\label{eq:KPSHhat}
\hat{K}_p(S,H)=\norm{\hat{\vec{\alpha}}}_p\max_i\norm{\hat{\vec{f}}_I(\vec{x}_i)}_{p^*},
\end{equation}
and $\hat{\vec{\alpha}}$  and  $\hat{\vec{f}}_I(\vec{x}_i)$ are the modified representation vector of $H$  and $\proj{\psi(x_i)}$ in the Pauli basis, respectively.
\end{thm}
The proof of this theorem is presented in Appendix  \ref{apen:single_unital}.

\subsection{Rademacher complexity of depth-\texorpdfstring{$l$}{l} quantum circuits}
In this subsection, we take the depth and width of the quantum circuits involved into account by considering the 
the layer structure of the circuits. Consider a depth-$l$ quantum circuit $C_l=\Phi_l\circ\Phi_{l-1}\circ\cdots\circ\Phi_1$,
where the $i$-th layer $\Phi_i: \mathcal{L}((\complex^2)^{\ot n_{i-1}})\to \mathcal{L}((\complex^2)^{\ot n_i})$ (see Fig.~\ref{fig0} for a circuit diagram). 
We shall denote the quantum circuit as $\vec{C}_l=(\Phi_l,\Phi_{l-1},...,\Phi_1)$ 
and the set of  quantum circuits with fixed depth $l$ and width vector $\vec{n}=(n_l,\ldots,n_1,n_0)$ as
\begin{align}
\mathcal{C}^{l,\vec{n}}
&=\Big\{\vec{C}_l\Big| \vec{C}_l=(\Phi_l,\Phi_{l-1},\ldots,\Phi_1), \nonumber\\
&\qquad\qquad
\Phi_i: \mathcal L((\complex^2)^{\ot n_{i-1}})\to \mathcal L((\complex^2)^{\ot n_i})\Big\}.
\end{align}

 \begin{figure}[!ht]
  \center{\includegraphics[width=8cm]  {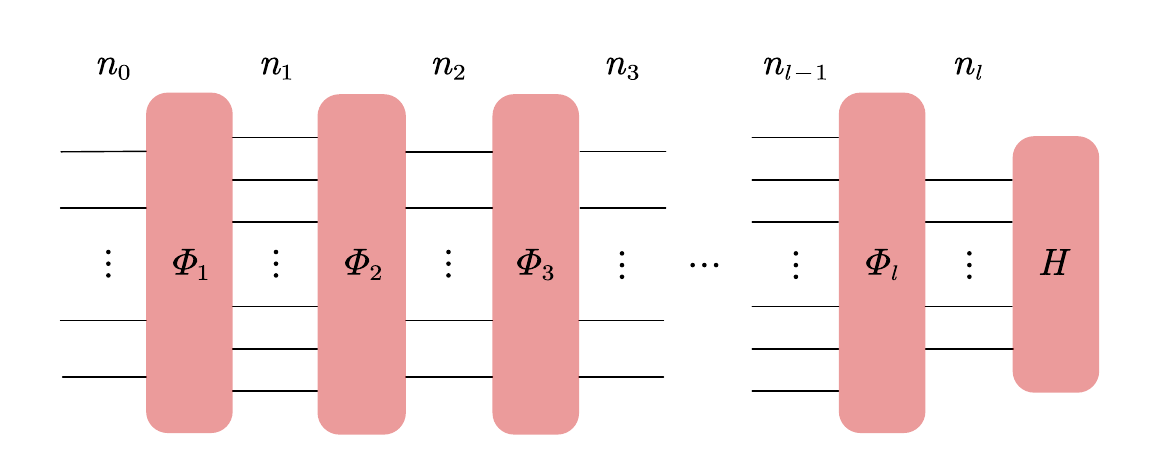}}     
  \caption{Circuit diagram of a depth-$l$ quantum circuit}
  \label{fig0}
 \end{figure}

Next, let us define the resource measure for a depth-$l$ quantum circuit $\vec{C}_l$ as follows:
\begin{eqnarray}
\nu_{p,q}(\vec{C}_l)=
 \frac{1}{l}\sum^l_{i=1}\norm{M^{\Phi_i}}_{p,q},
\end{eqnarray} 
which represents the average amount of  magic over  the layers of the quantum circuit. 
Let us denote $\mathcal{C}^{l,\vec{n}}_{\nu_{p,q}\leq \nu}$ to be the set of
quantum circuits with bounded resource $\nu_{p,q}\leq \nu$,  fixed depth $l$, and width vector $\vec{n}$ (See Fig.~\ref{fig2}). Then we have the following results. 
\begin{thm}\label{thm:Con_main1}
Given the set of depth-$l$ quantum circuits with bounded resource $\nu_{p,q}\leq \nu$,
the Rademacher complexity 
 on $m$ independent samples $S=\set{\vec{x_1},...,\vec{x}_m}$ is bounded as follows.
 
 (1) For $1\leq p\leq 2$, we have
\begin{eqnarray}
R_S(\mathcal{F}\circ \mathcal{C}^{l,\vec{n}}_{\nu_{p,q}\leq \nu})
\leq \nu^l 4^{\norm{\vec{n}}_1\max\set{\frac{1}{p^*},\frac{1}{q}}}
\frac{\sqrt{\min\set{p^*, 8n_0}}}{\sqrt{m}}K_p(S,H).\nonumber\\
\end{eqnarray}

(2) For $2<p<\infty$, we have
\begin{eqnarray}
R_S(\mathcal{F}\circ \mathcal{C}^{l,\vec{n}}_{\nu_{p,q}\leq \nu})
\leq \nu^l 4^{\norm{\vec{n}}_1\max\set{\frac{1}{p^*},\frac{1}{q}}}
\frac{\sqrt{p^*}}{m^{1/p}}K_p(S,H),
\end{eqnarray}
where $K_p(S,H)$ is defined by Eq.~\eqref{eq:KpSH}, and $\norm{\vec{n}}_1=\sum^l_{i=1}n_i$.
\end{thm}
This theorem tells us how the Rademacher complexity depends on the depth, width and the amount of magic in the quantum circuits. Note that we can choose suitable $p,q$ to reduce the exponential dependence on the width vector to polynomial dependence, for example, by taking $p^*=q=\Omega(\norm{\vec{n}}_1/\log\norm{\vec{n}}_1)$ or $p^*=q=\infty$.
The proof of Theorem \ref{thm:Con_main1} is presented in Appendix \ref{apen:deep}.

 \begin{figure}[!ht]
  \center{\includegraphics[width=6cm]  {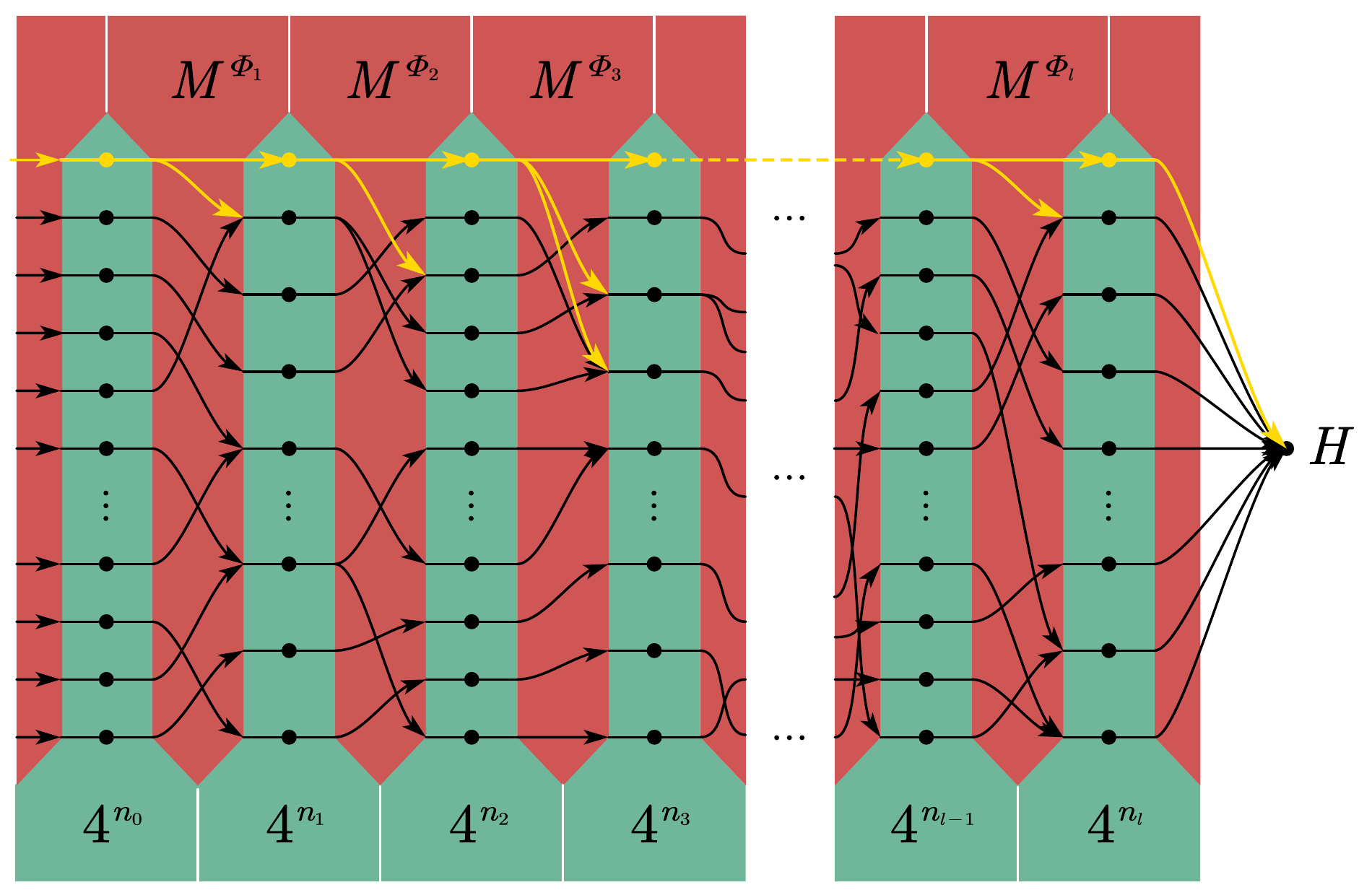}}     
  \caption{Diagram illustrating the layer structure of the representation matrix of a depth-$l$ quantum circuit}
  \label{fig2}
 \end{figure}

If the quantum channel in the quantum circuit is unital (for example, a unitary quantum channel), then we modify the 
resource measure as follows (See Fig. \ref{fig3}):
\begin{eqnarray}
\hat{\nu}_{p,q}(\vec{C}_l)=
\frac{1}{l}\sum^l_{i=1}\norm{\hat{M}^{\Phi_i}}_{p,q}.
\end{eqnarray}
We are now ready to state our next result.

\begin{thm}\label{thm:Con_main2}
Let $H$ be a traceless observable. Given the set of depth-$l$ quantum circuits with bounded resource  $\hat{\nu}_{p,q}\leq \nu$,
  the Rademacher complexity of $\mathcal{F}\circ \mathcal{C}^{l,\vec{n}}_{\hat{\nu}_{p,q}\leq \nu}$
 on $m$ independent samples $S=\set{\vec{x_1},...,\vec{x}_m}$ satisfies the following bounds 
 
 (1) For $1\leq p\leq 2$, we have
\begin{eqnarray}
R_S(\mathcal{F}\circ \mathcal{C}^{l,\vec{n}}_{\hat{\nu}_{p,q}\leq \nu})
\leq \nu^l \Pi^l_{i=1}N^{\max\set{\frac{1}{p^*},\frac{1}{q}}}_i
\frac{\sqrt{\min\set{p^*, 8n_0}}}{\sqrt{m}}\hat{K}_p(S,H).
\nonumber\\
\end{eqnarray}

(2) For $2<p<\infty$, we have
\begin{eqnarray}
R_S(\mathcal{F}\circ \mathcal{C}^{l,\vec{n}}_{\hat{\nu}_{p,q}\leq \nu})
\leq \nu^l \Pi^l_{i=1}N^{\max\set{\frac{1}{p^*},\frac{1}{q}}}_i\frac{\sqrt{p^*}}{m^{1/p}}\hat{K}_p(S,H).
\end{eqnarray}
where $N_i=4^{n_i}-1$ for any $1\leq i\leq l$ and $\hat{K}_p(S,H)$ is defined by Eq.~\eqref{eq:KPSHhat}.
\end{thm}
The proof of this theorem is presented in Appendix \ref{apen:deep_unital}.

 \begin{figure}[!ht]
  \center{\includegraphics[width=6cm]  {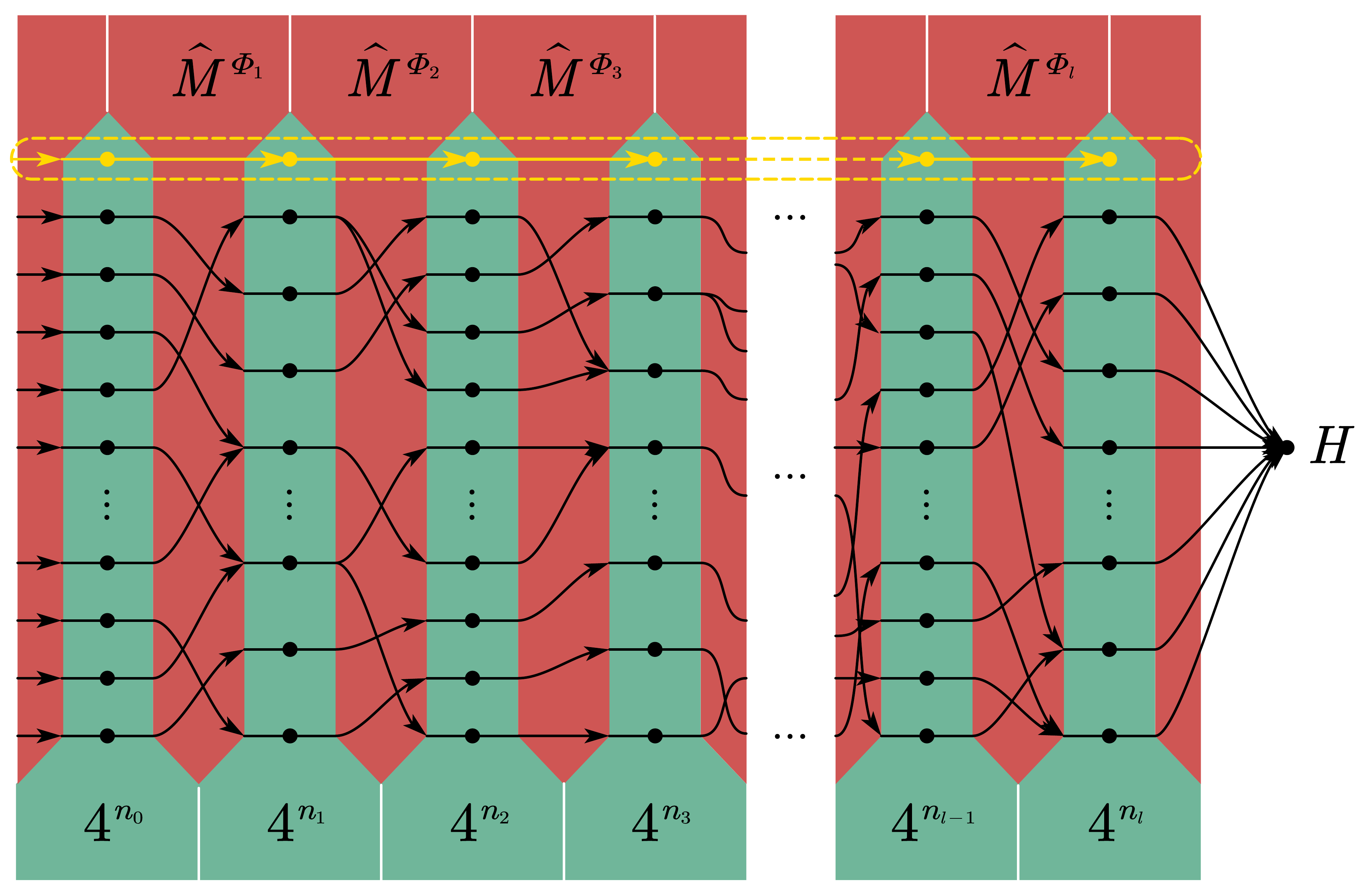}}     
  \caption{Diagram illustrating 
  the layer structure of the representation matrix of a depth-$l$ unital quantum circuit.
  }
  \label{fig3}
 \end{figure}

\begin{remark}
While we based our resource measure in this paper on the arithmetic mean, we could have alternatively defined a resource measure of the quantum circuit
$\vec{C}_l=(\Phi_l,\Phi_{l-1},\ldots,\Phi_1)$
that is based on the geometric mean, viz.
\begin{eqnarray}
\mu_{p,q}(\vec{C}_l)
=\prod^{l}_{i=1}\norm{M^{\Phi_i}}_{p,q},
\end{eqnarray}
which is the geometric mean of the resource over the layers
of the quantum circuit. By the arithmetic mean--geometric mean inequality, it is easy to 
see that 
\begin{eqnarray}
\nu_{p,q}(\vec{C}_l)
\geq \mu_{p,q}(\vec{C}_l)^{1/l}.
\end{eqnarray}
Also, we could define the path norm as a resource measure
as follows:
\begin{eqnarray}
\gamma_{p,q}(\vec{C}_l)
=\left(\frac{1}{4^{n_l}}\sum_{\vec{x}}\gamma^{(\vec{x})}_p(C_l)^q\right)^{1/q},
\end{eqnarray}
where
\begin{eqnarray}
\gamma^{(\vec{x})}_{p}(\vec{C}_l)
=\left(
\sum_{\substack{v_0\to v_1\to\ldots\to v_{out},\\
v_{out}=\vec{x}}}
\left|M^{\Phi_l}_{\vec{x}v_{l-1}}M^{\Phi_{l-1}}_{v_{l-1}v_{l-2}}\cdots M^{\Phi_{1}}_{v_{1}v_{0}} \right|^p
\right)^{1/p}.
\nonumber\\
\end{eqnarray}
The modified version of these resource measures for quantum circuits can also be similarly defined. 
We present similar results on the Rademacher complexity of quantum circuits based on these resource measures in Appendices \ref{apen:deep} and \ref{apen:deep_unital}. 
\end{remark}

\begin{remark} 
Note that for any given quantum channel $\Psi$, 
there could be many different ways  to realize it by quantum circuits of the same depth 
$l$ and width vector $\vec{n}$, i.e., there could be multiple circuits
$\vec{C}_l=(\Phi_l,\ldots,\Phi_1)$ for which $\Psi=\Phi_l\circ\ldots\circ\Phi_1$. Furthermore, note that resource measures such as $\nu_{p,q}$ depend on the realization of the channel. Thence, if we would like to define a resource measure for quantum channels $\Psi$ that is independent of their quantum circuit realization, it would be necessary to adopt a definition like the one below:
\begin{eqnarray}
\nu^{l,\vec{n}}_{p,q}(\Psi)
:=\min\set{\nu_{p,q}(\vec{C}_l): 
\vec{C}_l\in\mathcal{C}^{l,\vec{n}},
\Psi=C_l
},
\end{eqnarray}
which quantifies the minimum amount of resources necessary to realize the target channel over all quantum circuits with a
given depth and width. The quantities 
$\mu^{l,\vec{n}}_{p,q}$ and $\gamma^{l,\vec{n}}_{p,q}$ may also be defined analogously. These resource
measures may be of independent interest in resource theory. 
\end{remark}

\section{Conclusion}

In this work, we studied the Rademacher complexity of quantum circuits. First, 
we introduced the $(p,q)$ group norm to define the resource measure of magic for 
quantum channels and for quantum circuits with a layered structure.
Second, we proved that the Rademacher complexity of quantum circuits is bounded by its depth and width as well as its amount of magic, where 
the dependence on the width is determined by the choice of 
$(p,q)$. These results reveal the dependence of statistical complexity on the resources and structure parameters (such as depth and width) of the quantum circuit.

While our results are stated in terms of the Rademacher complexity, there are other prominent choices of measures of statistical complexity, such as the VC dimension and metric entropy, that could be used. 
Due to the close relationship between the Rademacher complexity and the VC dimension and the metric entropy \cite{Dudley67,Sudakov71,Mendelson03}, it 
is straightforward to extend our results to obtain bounds on these complexity measures of quantum circuits. Another measure that has recently gained prominence is the topological entropy,
a concept from dynamic systems that has recently been used to measure the complexity of classical neural networks \cite{bu2020depth}. We leave for future work the problem of generalizing the results about Rademacher complexity to 
topological entropy. 
Finally, we note that while our results are based on expressing each quantum channel in
the Pauli basis, there are also other choices of bases, or more generally frames, that can be used to express  quantum channels, a notable example being the phase space point operator basis \cite{Veitch14, ferrie2009framed}. How do our results generalize to the case where the basis is chosen arbitrarily? We leave this question for further work.

\begin{acknowledgments}
K. B. thanks Arthur Jaffe  and Zhengwei Liu for the help and support during the breakout of the COVID   -19 pandemic. 
K. B. acknowledges the support of
 ARO Grants W911NF-19-1-0302 and
W911NF-20-1-0082, and the support from Yau Mathematical
Science Center at Tsinghua University during the visit. 

\end{acknowledgments}

 \bibliography{SatCom-lit}

\appendix
\widetext

\section{Single quantum channels}\label{apen:single}

\subsection{\texorpdfstring{$(p,q)$}{(p,q)} group norm of the representation matrix of a single quantum channel}
For any $N_1\times N_2$ real-valued matrix $M$, 
which can be written as a column
\begin{equation}
\left(\begin{array}{cc}
M_1\\
M_2\\
\ldots\\
M_{N_1}
\end{array}
\right)
\end{equation}
 of $N_1$ rows, we define the $(p,q)$ \textit{group norm}, with $0<p,q\leq \infty$, as follows:
\begin{eqnarray}
\norm{M}_{p,q}=
\left(\frac{1}{N_1}\sum_{i}\norm{M_i}^q_p\right)^{1/q},
\end{eqnarray}
where  the $l_p$ norm of the $i$-th row vector $M_i$ is 
\begin{eqnarray}
\norm{M_i}_p
=\left(\sum^{N_2}_{j=1}
|M_{ij}|^p\right)^{1/p}.
\end{eqnarray}

The $(p,q)$ group norm satisfies the following multiplicative property.
\begin{lem}\label{lem:multi_pq}
Given two matrices $M_1$ and $M_2$, it holds that
\begin{eqnarray}
\norm{M_1\otimes M_2}_{p,q}
=\norm{M_1}_{p,q}\norm{M_2}_{p,q}.
\end{eqnarray}
\end{lem}
\begin{proof}
This follows directly from the fact that  $[M_1\otimes M_2]_{\vec{x_1}\vec{x_2}, \vec{y}_1\vec{y}_2}
=[M_1]_{\vec{x}_1\vec{y}_1}[M_2]_{\vec{x_2}\vec{y}_2}$.
\end{proof}

Let $P_0 = \mathbb I$, $P_1 = X$, $P_2 = Y$, and $P_3 = Z$ be the single-qubit Pauli matrices. The $n$-qubit Pauli matrices $P_{\vec{z}}$ are defined as
$P_{\vec{z}}=P_{z_1}\ot P_{z_2}\ot \ldots\ot P_{z_n}$ for any vector $\vec{z}\in \set{0,1,2,3}^n$. 
Given a quantum channel $\Phi:\mathcal L((\complex^2)^{\ot n_1})\to \mathcal L((\complex^2)^{\ot n_2})$ from $n_1$ qubits to $n_2$ qubits, we define the $4^{n_1}\times 4^{n_2}$ representation matrix  $M^{\Phi}$ in the Pauli basis by its matrix elements as follows:
\begin{eqnarray}
M^{\Phi}_{\vec{z}\vec{x}}
=\frac{1}{2^{n_2}}\Tr{P_{\vec{z}}\Phi(P_{\vec{x}})},
\end{eqnarray}
where $\vec{x}\in \set{0,1,2,3}^{n_1}$, $\vec{z}\in \set{0,1,2,3}^{n_2}$, and $P_{\vec{x}}$ and $P_{\vec{z}}$ are the corresponding Pauli operators. 
From the definition of $M^{\Phi}$, it is easy to see that 
the representation matrix of quantum channels in the Pauli basis
 satisfies the following  properties.

\begin{lem}\label{lem:prop_pq_u}
Given two quantum channels $\Phi_1:\mathcal L((\complex^2)^{\ot n_1})\to \mathcal L((\complex^2)^{\ot n_2})$ and
$\Phi_2:\mathcal L((\complex^2)^{\ot n_3})\to \mathcal L((\complex^2)^{\ot n_4})$, we have
\begin{eqnarray}
M^{\Phi_2\circ \Phi_1}
&=&M^{\Phi_2}M^{\Phi_1},\\
M^{\Phi_2\ot\Phi_1}
&=&M^{\Phi_2}\ot M^{\Phi_1},\\
M^{\lambda\Phi_1+\mu\Phi_2}
&=&\lambda M^{\Phi_1}
+\mu M^{\Phi_2}, \forall \lambda,\mu\in\real.
\end{eqnarray}

\end{lem}
\begin{proof}
Based on definition of the representation matrix $M^{\Phi}$, we have 
\begin{eqnarray*}
\Phi_1(P_{\vec{x}})
=\sum_{\vec{y}}M^{\Phi_1}_{\vec{y}\vec{x}}P_{\vec{y}}.
\end{eqnarray*}
Therefore, it follows that 
\begin{eqnarray*}
M^{\Phi_2\circ\Phi_1}_{\vec{z}\vec{x}}
=\frac{1}{2^n}
\Tr{P_{\vec{z}}\Phi_2\circ\Phi_1(P_{\vec{x}})}
=\frac{1}{2^{n_4}}
\Tr{P_{\vec{z}}\sum_{\vec{y}}M^{\Phi_1}_{\vec{y}\vec{x}}\Phi_2(P_{\vec{y}})}
=\sum_{\vec{y}}M^{\Phi_2}_{\vec{z}\vec{y}}M^{\Phi_1}_{\vec{y}\vec{x}}.
\end{eqnarray*}
Hence,
\begin{eqnarray*}
M^{\Phi_2\circ \Phi_1}
=M^{\Phi_2}M^{\Phi_1}.
\end{eqnarray*}

The other two identities
\begin{eqnarray*}
M^{\Phi_2\ot\Phi_1}
&=&M^{\Phi_2}\ot M^{\Phi_1},\\
M^{\lambda\Phi_1+\mu\Phi_2}
&=&\lambda M^{\Phi_1}
+\mu M^{\Phi_2},
\end{eqnarray*}
follow directly from the definition of the representation matrix.
\end{proof}

Next, we consider the $(p,q)$ group norm of the representation matrix  for 
quantum channels, for the case when the channel is unitary.

\begin{lem}\label{lem:cliff_M}
For a Clifford unitary $U$, $M^{U}$ is a permutation matrix, up to a $\pm $ sign. That is, 
$M^{U}_{\vec{x}\vec{y}}=\pm \delta_{\vec{x},\pi(\vec{y})}$, where $\pi$ is a permutation over 
the set 
$\set{0,1,2,3}^n$.
\end{lem}
\begin{proof}
This comes directly from the definition of the Clifford unitaries, which map Pauli 
operators to Pauli operators. 
\end{proof}

\begin{lem}
The $(p,q)$ group norm is invariant under the 
left or right multiplication of $M^{U}$, if $U$ is a Clifford unitary.
That is, for any matrix $A$, we have
\begin{eqnarray}
\norm{M^{U}A}_{p,q}
=\norm{AM^{U}}_{p,q}
=\norm{A}_{p,q}.
\end{eqnarray}
\end{lem}
\begin{proof}
Based on the  Lemma \ref{lem:cliff_M}, $M^{U}$ is a permutation matrix up to some $\pm$ sign.
Hence, 
$AM^{U}$ is just a permutation of the columns of $A$ with some $\pm $ sign. 
Thus, for the $i$-th  row vector, we have $\norm{(AM^{U})_i}_p=\norm{A_i}_p$. Therefore, 
$\norm{AM^{U}}_{p,q}=\norm{A}_{p,q}$.

Similarly, $M^UA$ is just a permutation of the columns of $A$ with some $\pm$ sign. 
Thus, for the $i$-th row vector 
$\norm{(M^UA)_i}_p=\norm{A_{\pi(i)}}_p$, where $\pi$ is a permutation. 
Then $\sum_{i}\norm{(AM^U)_i}_p=\sum_{i}\norm{A_{\pi(i)}}_p=\sum_{i}\norm{A_{i}}_p$, i.e., 
$
\norm{AM^{U}}_{p,q}
=\norm{A}_{p,q}$.

\end{proof}

\begin{lem}\label{lem:clif_pq}
Given a unitary channel $U$, we have the following result:

(1) For $0<p<2$, we have $\norm{M^{U}}_{p,q}\geq 1$, $\norm{M^{U}}_{p,q}= 1$ iff 
$U$ is a Clifford unitary. 

(2) For $p>2$, $0<q<\infty$, we have $\norm{M^{U}}_{p,q}\leq 1$, $\norm{M^{U}}_{p,q}= 1$ iff 
$U$ is a Clifford unitary. 

(3) For $p=2$, $q>0$ or $p>2, q=\infty$, we have $\norm{M^{U}}_{p,q}=1$ for any unitary $U$.
\end{lem}

\begin{proof}
First, for any unitary $U$, it is easy to see that  $M^{U}$ is an orthogonal matrix. Therefore, 
$\norm{M^U_{\vec{x}}}_2=1$ for any  $\vec{x}$ and $M^U_{\vec{0}}=(1,0,\ldots,0)$. Therefore, we have the statement in (3).

(1) For $0<p<2$, we have 
$\norm{M^U_{\vec{x}}}_p\geq \norm{M^U_{\vec{x}}}_2=1$ for any $\vec{x}$. 
Therefore, $\norm{M^{U}}_{p,q}\geq 1$. Besides, 
$\norm{M^{U}}_{p,q}= 1$ iff $\norm{M^U_{\vec{x}}}_p=1$ for any $\vec{x}$
iff every  row vector  $M^{U}_{\vec{x}}$ has only one nonzero element, which could only be $\pm 1$,
iff $U$ is a Clifford unitary.

(2) For $p>2$, $0<q<\infty$, we have 
$\norm{M^U_{\vec{x}}}_p\leq \norm{M^U_{\vec{x}}}_2=1$ for any $\vec{x}$. 
Therefore, $\norm{M^{U}}_{p,q}\leq 1$. Besides, 
$\norm{M^{U}}_{p,q}= 1$ iff $\norm{M^U_{\vec{x}}}_p=1$ for any $\vec{x}$
iff every  row  vector $M^{U}_{\vec{x}}$ has only one nonzero element, which could only be $\pm 1$,
iff $U$ is Clifford.
\end{proof}

Based on the above facts, it is easy to see that the $(p,q)$ norm of the representation matrix can be regarded as some 
resource measure of magic of quantum gates.

\begin{prop}\label{prop:clif_pq1}
Given a unitary channel $U$, the $(p,q)$ norm can be regarded as a 
resource measure satisfying the following properties

(1) (Faithfulness) For $0<p<2$, we have $\norm{M^{U}}_{p,q}\geq 1$, $\norm{M^{U}}_{p,q}= 1$ iff 
$U$ is Clifford unitary. 

(1') (Faithfulness) For $p>2$, $0<q<\infty$, we have $\norm{M^{U}}_{p,q}\leq 1$, $\norm{M^{U}}_{p,q}= 1$ iff 
$U$ is Cllifford unitary.

(2) (Invariance under Clifford unitaries)  $\norm{M^{U_1\circ U\circ U_2}}_{p,q}=\norm{M^{ U}}_{p,q}$ for any 
Clifford unitaries $U_1$ and $U_2$.

(3) (Multiplicity under tensor product) $\norm{M^{U_1 \ot U_2}}_{p,q}=\norm{M^{ U_1}}_{p,q}\norm{M^{ U_2}}_{p,q} $.

(4) (Convexity) For $p\geq 1, q\geq 1$, we have $\norm{M^{\lambda U_1 +(1-\lambda) U_2}}_{p,q}\leq \lambda\norm{M^{ U_1}}_{p,q}+(1-\lambda)\norm{M^{ U_2}}_{p,q}$ for $\lambda\in[0,1]$.

\end{prop}
\begin{proof}
(1) and (1') come from  Lemma \ref{lem:clif_pq} directly.

(2) \begin{eqnarray}
\norm{M^{U_1\circ U\circ U_2}}_{p,q}=
\norm{M^{U_1}M^{U}M^{U_2}}_{p,q}
=\norm{M^{ U}}_{p,q},
\end{eqnarray}
where the first equality comes from Lemma \ref{lem:prop_pq_u} and the second equality comes from Lemma \ref{lem:clif_pq}.

(3)  \begin{eqnarray}
\norm{M^{U_1 \ot U_2}}_{p,q}=
\norm{M^{U_1} \ot M^{U_2}}_{p,q}
=\norm{M^{ U_1}}_{p,q}\norm{M^{ U_2}}_{p,q},
\end{eqnarray}
where the first equality comes from Lemma \ref{lem:prop_pq_u} and the second equality comes from Lemma \ref{lem:multi_pq}.

(4) comes directly from the convexity of $l_p$ and $l_q$ norm for $p\geq 1,q\geq 1$.

\end{proof}

\subsection{Bounds on the Rademacher complexity of quantum channels}

Let $p^*$ denote the H\"{o}lder conjugate of $p$, i.e., $\frac{1}{p}+\frac{1}{p^*}=1$.

\begin{lem}\label{lem:normpRaN1}
For any $N_1\times N_2$ real-valued matrix $M$, and any vector $\vec{v}\in\real^{N_2}$, we have 
\begin{eqnarray}\label{eq:normpRaN1}
\norm{M\vec{v}}_{p^*}
\leq N^{\max\set{\frac{1}{p^*},\frac{1}{q}}}_1
\norm{M}_{p,q}\norm{\vec{v}}_{p^*}.
\end{eqnarray}
\end{lem}
\begin{proof}
First, let us prove the following 
inequality
\begin{eqnarray*}
\norm{M\vec{v}}_{p^*}
\leq N^{\frac{1}{p^*}}_1\norm{M}_{p,p^*}
\norm{\vec{v}}_{p^*}.
\end{eqnarray*}
This inequality holds because
\begin{eqnarray*}
\norm{M\vec{v}}_{p^*}^{p^*}
=\sum_{i}(M_i\vec{v})^{p^*}
\leq \sum_{i}\norm{M_i}^{p^*}_p\norm{\vec{v}}^{p^*}_{p^*}
=N^{\frac{1}{p^*}}_1\norm{M}_{p,p^*}
\norm{\vec{v}}_{p^*}.
\end{eqnarray*}

If $q>p^*$, then $\max\set{\frac{1}{p^*},\frac{1}{q}}=\frac{1}{p^*}$ and $\norm{M}_{p,q}\geq \norm{M}_{p,p^*}$. Hence the inequality Eq.~\eqref{eq:normpRaN1} reduces
to
\begin{eqnarray*}
\norm{M\vec{v}}_{p^*}
\leq N^{\frac{1}{p^*}}_1\norm{M}_{p,p^*}
\norm{\vec{v}}_{p^*}.
\end{eqnarray*}

If $q<p^*$, then $\max\set{\frac{1}{p^*},\frac{1}{q}}=\frac{1}{q}$ and $N^{1/q}_1\norm{M}_{p,q}\geq N^{1/p^*}_1\norm{M}_{p,p^*}$. Hence the inequality Eq.~\eqref{eq:normpRaN1}  
reduces to 
\begin{eqnarray*}
\norm{M\vec{v}}_{p^*}
\leq N^{\frac{1}{p^*}}_1\norm{M}_{p,p^*}
\norm{\vec{v}}_{p^*}.
\end{eqnarray*}

\end{proof}

\begin{lem}\label{lem:depn1}
For any $1\leq p\leq 2$
\begin{eqnarray}
\mathbb{E}_{\vec{\epsilon}}
\frac{1}{m}
\norm{\sum^m_{i=1}\epsilon_i\vec{v}_i}_{p^*}
\leq \frac{\sqrt{\min\set{p^*, 8n_0}}}{\sqrt{m}}\max_i\norm{\vec{v}_i}_{p^*}.
\end{eqnarray}
For $2<p<\infty$, we have 
\begin{eqnarray}
\mathbb{E}_{\vec{\epsilon}}
\frac{1}{m}
\norm{\sum^m_{i=1}\epsilon_i\vec{v}_i}_{p^*}
\leq \frac{\sqrt{p^*}}{m^{1/p*}}\max_i\norm{\vec{v}_i}_{p^*},
\end{eqnarray}
where $\vec{v}_i\in \real^N$. 
\end{lem}
\begin{proof}
The proof is similar to that of Lemma 15 in \cite{Neyshabur15}. 
If
$1\leq p\leq \frac{2\log_2( N)}{2\log_2(N)-1}$, then $2\log_2 (N)\leq p^*$.
Thence,

\begin{eqnarray*}
\mathbb{E}_{\vec{\epsilon}}
\frac{1}{m}
\norm{\sum^m_{i=1}\epsilon_i\vec{v}_i}_{p^*}
&\leq& N^{\frac{1}{p^*}}
\mathbb{E}_{\vec{\epsilon}}
\frac{1}{m}
\norm{\sum^m_{i=1}\epsilon_i\vec{v}_i}_{\infty}\\
&\leq& N^{\frac{1}{2\log_2(N)}}
\mathbb{E}_{\vec{\epsilon}}
\frac{1}{m}
\norm{\sum^m_{i=1}\epsilon_i\vec{v}_i}_{\infty}\\
&\leq& \sqrt{2}
\mathbb{E}_{\vec{\epsilon}}
\frac{1}{m}
\norm{\sum^m_{i=1}\epsilon_i\vec{v}_i}_{\infty}\\
&=&\sqrt{2}
\mathbb{E}_{\vec{\epsilon}}
\frac{1}{m}\max_{j}
|\sum^m_{i=1}\epsilon_iv_i(j)|\\
&\leq&\sqrt{2}\frac{\sqrt{2\log(N)}}{m}
\max_{j}\norm{(v_i(j))_i}_2\\
&\leq&\sqrt{2}\frac{\sqrt{2\log(N)}}{\sqrt{m}}\max_i\norm{\vec{v}_i}_\infty\\
&\leq&\sqrt{2}\frac{\sqrt{2\log(N)}}{\sqrt{m}}\max_i\norm{\vec{v}_i}_{p^*}.
\end{eqnarray*}

If $  \frac{2\log_2(N)}{2\log_2(N)-1}<p<\infty$, then by
the Khintchine-Kahane inequality, we have 
\begin{eqnarray*}
\mathbb{E}_{\vec{\epsilon}}
\frac{1}{m}
\norm{\sum^m_{i=1}\epsilon_i\vec{v}_i}_{p^*}
\leq 
\frac{1}{m}
\left(
\sum_{j}\mathbb{E}_{\vec{\epsilon}}|\sum_i\epsilon_i
v_i(j)|^{p*}
\right)^{\frac{1}{p^*}}
\leq \frac{\sqrt{p^*}}{m}
\left(
\sum_{j}\norm{
(v_i(j))_i}^{p*}_2
\right)^{\frac{1}{p^*}},
\end{eqnarray*}
where 
\begin{equation}
\left(
\sum_{\vec{z}}\norm{
(v_i(j))_i}^{p*}_2
\right)^{\frac{1}{p^*}}
\leq
\left\{
\begin{array}{cc}
m^{1/2}\max_i \norm{\vec{v}_i}_{p^*}, p^*\geq 2, \\
m^{1/p^*}\max_i \norm{\vec{v}_i}_{p^*}, p^*< 2,
\end{array}
\right.
\end{equation}
and the first inequality comes from the Minkowski inequality 
 and the second inequality
from the fact that 
\begin{eqnarray*}
(x+y)^{p^*/2}\leq x^{p^*/2}+y^{p^*/2},
\end{eqnarray*}
for $p^*/2<1$.
Therefore 
\begin{equation}
\mathbb{E}_{\vec{\epsilon}}
\frac{1}{m}
\norm{\sum^m_{i=1}\epsilon_i\vec{v}_i}_{p^*}
\leq
\left\{
\begin{array}{cc}
\frac{\sqrt{p^*}}{m^{1/2}}\max_i \norm{\vec{v}_i}_{p^*}, p^*\geq 2,\\
\frac{\sqrt{p^*}}{m^{1/p}}\max_i \norm{\vec{v}_i}_{p^*}, p^*< 2.
\end{array}
\right.
\end{equation}

\end{proof}

\begin{lem}[Massart lemma \cite{shalev2014}]
Given a finite set $A\subset \real^m$, we have 
\begin{eqnarray}
R(A)\leq \max_{\vec{v}\in A}\norm{\vec{v}-\bar{\vec{v}}}_2\frac{\sqrt{2\log |A|}}{m},
\end{eqnarray}
where $\bar{\vec{v}}=\frac{1}{|A|}\sum_{\vec{v}\in A}\vec{v}$. 
\end{lem}

\begin{thm}[Restatement of Theorem \ref{thm:one}]
Given a set of  quantum circuits $\Phi$ from $n_0$ qubits to $n_1$ qubits with bounded $(p,q)$ norm $\norm{\cdot}_{p,q}$,
 the Rademacher complexity 
 on $m$ samples $S=\set{\vec{x_1},...,\vec{x}_m}$ satisfies the following bounds 
 
 (1) For $1\leq p\leq 2$, we have
\begin{eqnarray}
R_S(\mathcal{F}\circ \mathcal{C}_{\norm{\cdot}_{p,q}\leq \mu})
\leq \mu 4^{n_1\max\set{\frac{1}{p^*},\frac{1}{q}}}
\frac{\sqrt{\min\set{p^*, 8n_0}}}{\sqrt{m}}\norm{\vec{\alpha}}_p\max_i\norm{\vec{f}_I(\vec{x}_i)}_{p^*}.
\end{eqnarray}

(2) For $2<p<\infty$, we have
\begin{eqnarray}
R_S(\mathcal{F}\circ\mathcal{C}_{\norm{\cdot}_{p,q}\leq \mu})
\leq \mu 4^{n_1\max\set{\frac{1}{p^*},\frac{1}{q}}}
\frac{\sqrt{p^*}}{m^{1/p}}\norm{\vec{\alpha}}_p\max_i\norm{\vec{f}_I(\vec{x}_i)}_{p^*}.
\end{eqnarray}

\end{thm}
\begin{proof}
First, we compute
\begin{eqnarray*}
R_S(\mathcal{F}\circ \mathcal{C}_{\norm{\cdot}_{p,q}\leq \mu})
&=&
\mathbb{E}_{\vec{\epsilon}}
\frac{1}{m}\sup_{\Phi\in \mathcal{C}_{\norm{\cdot}_{p,q}\leq \mu}}
\left| \sum^m_{i=1}\epsilon_i\vec{\alpha}\vec{f}_{\Phi}(\vec{x}_i)\right|\\
&\leq&
\mathbb{E}_{\vec{\epsilon}}
\frac{1}{m}\sup_{\Phi}
\frac{\mu}{\norm{M^\Phi}_{p,q}}
\left| \sum^m_{i=1}\epsilon_i\vec{\alpha}\vec{f}_{\Phi}(\vec{x}_i)\right|\\
&=&\mu\mathbb{E}_{\vec{\epsilon}}
\frac{1}{m}\sup_{\Phi}
\frac{1}{\norm{M^\Phi}_{p,q}}
\left| \sum^m_{i=1}\epsilon_i\vec{\alpha}\vec{f}_{\Phi}(\vec{x}_i)\right|\\
&\leq&
\mu\norm{\vec{\alpha}}_p\mathbb{E}_{\vec{\epsilon}}
\frac{1}{m}\sup_{\Phi}
\frac{1}{\norm{M^\Phi}_{p,q}}
\norm{ \sum^m_{i=1}\epsilon_i\vec{f}_{\Phi}(\vec{x}_i)}_{p^*}\\
&\leq&
\mu\norm{\vec{\alpha}}_p\mathbb{E}_{\vec{\epsilon}}
\frac{1}{m}\sup_{\Phi}
\frac{1}{\norm{M^\Phi}_{p,q}}
\norm{ \sum^m_{i=1}\epsilon_iM^{\Phi}\vec{f}_{I}(\vec{x}_i)}_{p^*}\\
&\leq&\mu\norm{\vec{\alpha}}_p N^{\max\set{\frac{1}{p^*},\frac{1}{q}}}_1\mathbb{E}_{\vec{\epsilon}}
\frac{1}{m}
\norm{\sum^m_{i=1}\epsilon_i\vec{f}_{I}(\vec{x}_i)}_{p^*},
\end{eqnarray*}
where the third inequality follows from Lemma \ref{lem:normpRaN1}.
Using Lemma \ref{lem:depn1}, we get the results of this theorem.

\end{proof}

\section{Single unital quantum channel}\label{apen:single_unital}

\subsection{\texorpdfstring{$(p,q)$}{(p,q)} group norm of the modified representation matrix of unital channels}

If a quantum channel $\Phi$ is unital, 
i.e., $\Phi(\mathbb{I})=\mathbb{I}$, 
then  the representation matrix $M^{\Phi}$ has the following form
\begin{equation}
M^{\Phi}=\Br{
\begin{array}{ccc}
1&\vec{0}^T\\
\vec{0}& \hat{M}^{\Phi}
\end{array}
}.
\end{equation}
We call $\hat{M}^\Phi$ the \textit{modified representation matrix} of $\Phi$. 
(Note that a unitary channel is a special case of a unital channel.)
 For a unital channel $\Phi$, we define  the $(p,q)$ \textit{group norm} of the modified
 representation matrix $\hat{M}^{\Phi}$ as follows:
 \begin{eqnarray}
\norm{\hat{M}^\Phi}_{p,q}
=\left(
\frac{1}{N}\sum_{\vec{x}\neq \vec{0}}
\left(\sum_{\vec{y}\neq \vec{0}}\left|M_{\vec{x},\vec{y}}\right|^p\right)^{\frac{q}{p}}
\right)^{\frac{1}{q}}
=\left(
\frac{1}{N}\sum_{\vec{x}\neq \vec{0}}
\norm{M^\Phi_{\vec{x}}}^q_p
\right)^{\frac{1}{q}},
\end{eqnarray}
where $N=4^n-1$.

We now state and prove the following properties of the $(p,q)$ norm of the representation matrix $\hat{M}^U$, where $U$ is a unitary channel. 

\begin{prop}
For any unitary channel $U$, we have the following relationships between
$\norm{M^{U}}_{p,q}$ and $\norm{\hat{M}^U}_{p,q}$, 

(1) For $0<p<2$, $0<q<\infty$ , we  have
\begin{eqnarray}
\norm{M^{U}}_{p,q}
\leq \norm{\hat{M}^U}_{p,q},
\end{eqnarray}
with equality iff $U$ is a Clifford unitary.

(2)
For $0<p<2$, $q=\infty $,  we have
\begin{eqnarray}
\norm{M^U}_{p,\infty}
=\norm{\hat{M}^U}_{p,\infty},
\end{eqnarray}
for any unitary $U$.

(3) For $p>2$, $q>0$, we have

\begin{eqnarray}
\norm{M^{U}}_{p,q}
\geq \norm{\hat{M}^U}_{p,q},
\end{eqnarray}
 with equality iff $U$ is a Clifford unitary
 
(4) For $p=2$, $q>0$
\begin{eqnarray}
\norm{M^{U}}_{p,q}
=\norm{\hat{M}^U}_{p,q}=1.
\end{eqnarray}
\end{prop}
\begin{proof}
(4) is obvious, as $M^U$ and $\hat{M}^U$ are orthogonal matrices. 
 
For $0<p<2$, $q>0$ we have
$\norm{M^U_{\vec{x}}}_p\geq \norm{M^U_{\vec{0}}}_p=1$ for any $\vec{x}\neq \vec{0}$. 
Therefore, $\norm{M^{U}}_{p,q}
\leq \norm{\hat{M}^U}_{p,q} $
for $0<q<\infty$ and 
$\norm{M^{U}}_{p,q}
= \norm{\hat{M}^U}_{p,q} $
for $q=\infty$.  Hence, we get (2).
Next, for $0<q<\infty$, 
 $\norm{M^{U}}_{p,q}
= \norm{\hat{M}^U}_{p,q} $
iff $\norm{M^U_{\vec{x}}}_p=1$ for any $\vec{x}\neq \vec{0}$
iff every  row vector  $M^{U}_{\vec{x}}$ has only one nonzero element, which could only  be $\pm 1$,
iff $U$ is a Clifford unitary. Hence, we get (1).

For $p>2$, $0<q\leq \infty$, we have 
$\norm{M^U_{\vec{x}}}_p\leq \norm{M^U_{\vec{0}}}_p=1$ for any $\vec{x}$. 
Therefore, $\norm{M^{U}}_{p,q}
\geq \norm{\hat{M}^U}_{p,q} $. Besides, 
$\norm{M^{U}}_{p,q}
= \norm{\hat{M}^U}_{p,q} $
 iff $\norm{M^U_{\vec{x}}}_p=1$ for any $\vec{x}$
iff every  row vector  $M^{U}_{\vec{x}}$ has only one nonzero element, which could only be $\pm 1$,
iff $U$ is a Clifford unitary. Therefore, we get (3).

\end{proof}

A direct consequence of the above proposition is the following corollary.
\begin{cor}\label{lem:hat_unitary}
Given a unitary channel $U$,  the $(p,q)$ group norm of the modified representation matrix $\hat{M}^U$ can be regarded as a 
resource measure which satisfies the following properties

(1) (Faithfulness) For $0<p<2$, $q>0$ we have $\norm{\hat{M}^{U}}_{p,q}\geq 1$, $\norm{\hat{M}^{U}}_{p,q}= 1$ iff 
$U$ is a Clifford unitary. 

(1') (Faithfulness) For $p>2$, $q>0$, we have $\norm{\hat{M}^{U}}_{p,q}\leq 1$, $\norm{\hat{M}^{U}}_{p,q}= 1$ iff 
$U$ is a Cllifford unitary.

(2) (Invariance under Clifford unitary)  $\norm{\hat{M}^{U_1\circ U\circ U_2}}_{p,q}=\norm{\hat{M}^{ U}}_{p,q}$ for any 
Clifford unitary $U_1$ and $U_2$.

(3) (Convexity) For $p\geq 1$, we have $\norm{\hat{M}^{\lambda U_1 +(1-\lambda) U_2}}_{p,q}\leq \lambda\norm{\hat{M}^{ U_1}}_{p,q}+(1-\lambda)\norm{\hat{M}^{ U_2}}_{p,q}$.

\end{cor}

\begin{prop}\label{prop:hat_unitary}
Let $U_1$ and $U_2$ be unitary channels.

(1) For $0<p<2$, $0<q<\infty$, we have 
\begin{eqnarray}
\norm{\hat{M}^{U_1}\ot \hat{M}^{U_2}}_{p,q}
\geq \norm{\hat{M}^{U_1\ot U_2}}_{p,q},
\end{eqnarray}
with equality iff $U_1$ and $U_2$ are Clifford unitaries. 

(2) For $0<p<2$, $q=\infty$, we have 
\begin{eqnarray}
\norm{\hat{M}^{U_1}\ot \hat{M}^{U_2}}_{p,\infty}
= \norm{\hat{M}^{U_1\ot U_2}}_{p,\infty},
\end{eqnarray}
for any unitaries $U_1$ and $U_2$.

(3) For $p>2$, $0<q\leq \infty$, we have 
\begin{eqnarray}
\norm{\hat{M}^{U_1}\ot \hat{M}^{U_2}}_{p,q}
\leq  \norm{\hat{M}^{U_1\ot U_2}}_{p,q}.
\end{eqnarray}
For $p>2$, $0<q<\infty$,  ``='' holds iff $U_1$ and $U_2$ are Clifford unitary.
\end{prop}

(4) For $p=2$, $0<q\leq \infty$, we have 
\begin{eqnarray}
\norm{\hat{M}^{U_1}\ot \hat{M}^{U_2}}_{2,q}
= \norm{\hat{M}^{U_1\ot U_2}}_{2,q}=1.
\end{eqnarray}

\begin{proof}

(3) is obvious as both $\hat{M}^{U_1}\ot \hat{M}^{U_2}$ and 
$ \hat{M}^{U_1\ot U_2}$ are orthogonal matrices.

Using the property
\begin{eqnarray*}
\norm{\hat{M}^{U_1}\ot \hat{M}^{U_2}}_{p,q}
=\norm{\hat{M}^{U_1}}_{p,q}\norm{\hat{M}^{U_2}}_{p,q},
\end{eqnarray*}
we find that for $0<q<\infty$,
\begin{eqnarray*}
\norm{\hat{M}^{U_1}\ot \hat{M}^{U_2}}^q_{p,q}
&=&\left(\frac{1}{N_1}\sum_{\vec{x}_1\neq \vec{0}} \norm{M^{U_1}_{\vec{x}_1}}^q_p\right)
\left(\frac{1}{N_2}\sum_{\vec{x}_2\neq \vec{0}} \norm{M^{U_2}_{\vec{x}_2}}^q_p\right)\\
\norm{\hat{M}^{U_1\ot U_2}}^q_{p,q}
&=&\left(\frac{1}{N_1N_2+N_1+N_2}\right)
\sum_{(\vec{x}_1,\vec{x}_2)\neq (\vec{0},\vec{0})}
\norm{M^{U_1}_{\vec{x}_1}}^q_p\norm{M^{U_2}_{\vec{x}_2}}^q_p,\\
&=&\left(\frac{1}{N_1N_2+N_1+N_2}\right)
\left(\sum_{\vec{x}_1\neq \vec{0},\vec{x}_2\neq\vec{0}}
\norm{M^{U_1}_{\vec{x}_1}}^q_p\norm{M^{U_2}_{\vec{x}_2}}^q_p
+\sum_{\vec{x}_1\neq \vec{0},\vec{x}_2=\vec{0}}
\norm{M^{U_1}_{\vec{x}_1}}^q_p\norm{M^{U_2}_{\vec{x}_2}}^q_p
+\sum_{\vec{x}_1=\vec{0},\vec{x}_2\neq \vec{0}}
\norm{M^{U_1}_{\vec{x}_1}}^q_p\norm{M^{U_2}_{\vec{x}_2}}^q_p
\right)\\
&=&\left(\frac{1}{N_1N_2+N_1+N_2}\right)
\left(\sum_{\vec{x}_1\neq \vec{0},\vec{x}_2\neq\vec{0}}
\norm{M^{U_1}_{\vec{x}_1}}^q_p\norm{M^{U_2}_{\vec{x}_2}}^q_p
+\sum_{\vec{x}_1\neq \vec{0}}
\norm{M^{U_1}_{\vec{x}_1}}^q_p
+\sum_{\vec{x}_2\neq \vec{0}}
\norm{M^{U_2}_{\vec{x}_2}}^q_p
\right),
\end{eqnarray*}
where $N_1=4^{n_1}-1, N_2=4^{n_2}-1$. 
Hence to compare $\norm{\hat{M}^{U_1}\ot \hat{M}^{U_2}}_{p,q}$ and $ \norm{\hat{M}^{U_1\ot U_2}}_{p,q}$, 
we need only to compare 

\begin{eqnarray*}
\left(\frac{1}{N_1}\sum_{\vec{x}_1\neq \vec{0}} \norm{M^{U_1}_{\vec{x}_1}}^q_p\right)
\left(\frac{1}{N_2}\sum_{\vec{x}_2\neq \vec{0}} \norm{M^{U_2}_{\vec{x}_2}}^q_p\right),
\end{eqnarray*} 
and 
\begin{eqnarray*}
\frac{1}{N_1+N_2}\left(\sum_{\vec{x}_1\neq \vec{0}}
\norm{M^{U_1}_{\vec{x}_1}}^q_p
+\sum_{\vec{x}_2\neq \vec{0}}
\norm{M^{U_2}_{\vec{x}_2}}^q_p
\right).
\end{eqnarray*}
To this end, let us consider a simple inequality first. 
It is easy to verify the following two inequalities:

(1) For $a,b\geq 1$, we have
\begin{eqnarray}
ab\geq \frac{N_1a+N_2b}{N_1+N_2}.
\end{eqnarray}
Morevover, equality holds iff $a=b=1$.

(2) For $0<a,b\leq1$, we have
\begin{eqnarray}
ab\leq\frac{N_1a+N_2b}{N_1+N_2}.
\end{eqnarray}
Moreover, equality holds iff $a=b=1$.

Thus, for $0<p<2$, $0<q<\infty$, we have $\norm{M^U_{\vec{x}}}^q_p\geq 1$ for any $\vec{x}\neq \vec{0}$, and 
 $\norm{M^U_{\vec{x}}}^q_p=1$ for all $\vec{x}\neq 0$ iff $U$ is Clifford. 
Let 
\begin{eqnarray*}
a=\frac{1}{N_1}\sum_{\vec{x}_1\neq \vec{0}} \norm{M^{U_1}_{\vec{x}_1}}^q_p,\\
b=\frac{1}{N_2}\sum_{\vec{x}_2\neq \vec{0}} \norm{M^{U_2}_{\vec{x}_2}}^q_p.
\end{eqnarray*}
 Then by the first inequality (1), we have 
\begin{eqnarray*}
\left(\frac{1}{N_1}\sum_{\vec{x}_1\neq \vec{0}} \norm{M^{U_1}_{\vec{x}_1}}^q_p\right)
\left(\frac{1}{N_2}\sum_{\vec{x}_2\neq \vec{0}} \norm{M^{U_2}_{\vec{x}_2}}^q_p\right)
\geq
\frac{1}{N_1+N_2}\left(\sum_{\vec{x}_1\neq \vec{0}}
\norm{M^{U_1}_{\vec{x}_1}}^q_p
+\sum_{\vec{x}_2\neq \vec{0}}
\norm{M^{U_2}_{\vec{x}_2}}^q_p
\right).
\end{eqnarray*}
Therefore, 
for $0<p<2$, $0<q<\infty$, we have 
\begin{eqnarray*}
\norm{\hat{M}^{U_1}\ot \hat{M}^{U_2}}_{p,q}
\geq \norm{\hat{M}^{U_1\ot U_2}}_{p,q},
\end{eqnarray*}
where equality holds iff $U_1$ and $U_2$ are Clifford unitary.

Similarly,  
for $p>2$, $0<q<\infty$, we have $\norm{M^U_{\vec{x}}}^q_p\leq 1$ for any $\vec{x}\neq \vec{0}$, and 
 $\norm{M^U_{\vec{x}}}^q_p=1$ for all $\vec{x}\neq 0$ iff $U$ is Clifford.

Let 
\begin{eqnarray*}
a=\frac{1}{N_1}\sum_{\vec{x}_1\neq \vec{0}} \norm{M^{U_1}_{\vec{x}_1}}^q_p,\\
b=\frac{1}{N_2}\sum_{\vec{x}_2\neq \vec{0}} \norm{M^{U_2}_{\vec{x}_2}}^q_p.
\end{eqnarray*}
 Then by the second inequality (2), we have 
\begin{eqnarray*}
\left(\frac{1}{N_1}\sum_{\vec{x}_1\neq \vec{0}} \norm{M^{U_1}_{\vec{x}_1}}^q_p\right)
\left(\frac{1}{N_2}\sum_{\vec{x}_2\neq \vec{0}} \norm{M^{U_2}_{\vec{x}_2}}^q_p\right)
\leq
\frac{1}{N_1+N_2}\left(\sum_{\vec{x}_1\neq \vec{0}}
\norm{M^{U_1}_{\vec{x}_1}}^q_p
+\sum_{\vec{x}_2\neq \vec{0}}
\norm{M^{U_2}_{\vec{x}_2}}^q_p
\right).
\end{eqnarray*}
Therefore,
for $p>2$, $0<q<\infty$, we have 
\begin{eqnarray*}
\norm{\hat{M}^{U_1}\ot \hat{M}^{U_2}}_{p,q}
\leq  \norm{\hat{M}^{U_1\ot U_2}}_{p,q}.
\end{eqnarray*}
Moreover, equality holds iff $U_1$ and $U_2$ are Clifford unitary.

 Now, let us consider the case where $q=\infty$. 
 For $q=\infty$, we have 
 \begin{eqnarray*}
\norm{\hat{M}^{U_1}\ot \hat{M}^{U_2}}_{p,\infty}
=\max_{\vec{x}_1\neq \vec{0}} \norm{M^{U_1}_{\vec{x}_1}}_p
\max_{\vec{x}_2\neq \vec{0}} \norm{M^{U_2}_{\vec{x}_2}}_p,
 \end{eqnarray*}
 and 
 \begin{eqnarray*}
 \norm{\hat{M}^{U_1\ot U_2}}_{p,\infty}
=\max_{(\vec{x}_1,\vec{x}_2)\neq (\vec{0},\vec{0})}
\norm{M^{U_1}_{\vec{x}_1}}_p\norm{M^{U_2}_{\vec{x}_2}}_p
=\max\left\{ \max_{\vec{x}_1\neq \vec{0}} \norm{M^{U_1}_{\vec{x}_1}}_p
\max_{\vec{x}_2\neq \vec{0}} \norm{M^{U_2}_{\vec{x}_2}}_p, \max_{\vec{x}_1\neq \vec{0}} \norm{M^{U_1}_{\vec{x}_1}}_p,
\max_{\vec{x}_2\neq \vec{0}}
\norm{M^{U_2}_{\vec{x}_2}}_p\right\}.
 \end{eqnarray*}
 Hence, 
 for $0<p<2$, we have 
$ \norm{M^{U}_{\vec{x}}}_p\geq 1$ for any $\vec{x}\neq \vec{0}$; therefore,
\begin{eqnarray*}
\max_{\vec{x}_1\neq \vec{0}} \norm{M^{U_1}_{\vec{x}_1}}_p\max_{\vec{x}_2\neq \vec{0}} \norm{M^{U_2}_{\vec{x}_2}}_p
\geq \max\left\{\max_{\vec{x}_1\neq \vec{0}} \norm{M^{U_1}_{\vec{x}_1}}_p,
\max_{\vec{x}_2\neq \vec{0}}
\norm{M^{U_2}_{\vec{x}_2}}_p\right\}.
\end{eqnarray*}
That is, 
 \begin{eqnarray*}
 \norm{\hat{M}^{U_1}\ot \hat{M}^{U_2}}_{p,\infty}
 =\norm{\hat{M}^{U_1\ot U_2}}_{p,\infty}.
 \end{eqnarray*}
 
For $p>2$, we have 
$ \norm{M^{U}_{\vec{x}}}_p\leq 1$ for any $\vec{x}\neq \vec{0}$; therefore,
\begin{eqnarray*}
\max_{\vec{x}_1\neq \vec{0}} \norm{M^{U_1}_{\vec{x}_1}}_p\max_{\vec{x}_2\neq \vec{0}} \norm{M^{U_2}_{\vec{x}_2}}_p
\leq \max\left\{\max_{\vec{x}_1\neq \vec{0}} \norm{M^{U_1}_{\vec{x}_1}}_p,
\max_{\vec{x}_2\neq \vec{0}}
\norm{M^{U_2}_{\vec{x}_2}}_p\right\}.
\end{eqnarray*}

\end{proof}

\subsection{Rademacher complexity of single unital quantum circuit}
In this subsection, we will assume for simplicity that the observable $H$ is traceless, which implies that
$\alpha_{\vec{0}}=0$.
\begin{thm}[Restatement of Theorem \ref{thm:two}]
Given the set of unital  quantum circuits $\Phi$ from $n_0$ qubits to $n_1$ qubits with bounded $(p,q)$ norm of the 
modified representation matrix,
 the Rademacher complexity 
 on $m$ samples $S=\set{\vec{x_1},...,\vec{x}_m}$ satisfies the following bounds 
 
 (1) For $1\leq p\leq 2$, we have
\begin{eqnarray}
R_S(\mathcal{F}\circ \mathcal{C}^{n_0,n_1}_{\norm{\hat{\cdot}}_{p,q}\leq \mu})
\leq \mu N^{\max\set{\frac{1}{p^*},\frac{1}{q}}}_1
\frac{\sqrt{\min\set{p^*, 8n_0}}}{\sqrt{m}}\norm{\hat{\vec{\alpha}}}_p\max_i\norm{\hat{\vec{f}}_I(\vec{x}_i)}_{p^*},
\end{eqnarray}
where $N_1=4^{n_1}-1$.

(2) For $2<p<\infty$, we have
\begin{eqnarray}
R_S(\mathcal{F}\circ\mathcal{C}^{n_0,n_1}_{\norm{\hat{\cdot}}\leq \mu})
\leq \mu N^{\max\set{\frac{1}{p^*},\frac{1}{q}}}_1
\frac{\sqrt{p^*}}{m^{1/p}}\norm{\hat{\vec{\alpha}}}_p\max_i\norm{\hat{\vec{f}}_I(\vec{x}_i)}_{p^*}.
\end{eqnarray}

\end{thm}
\begin{proof}
Since  $\alpha_{\vec{x}}=0$, it follows that $\vec{\alpha}=(0,\hat{\vec{\alpha}})$,  where $\hat{\vec{\alpha}}\in \real^{N_1}$. Hence, 
\begin{eqnarray*}
f_{\Phi}(\vec{x})
&=&\Tr{\Phi(\proj{\psi(\vec{x})})H}\\
&=&\sum_{\vec{z}\neq \vec{0}}\alpha_{\vec{z}} \Tr{\Phi(\proj{\psi(\vec{x}_i)})P_{\vec{z}}}\\
&=&\hat{\vec{\alpha}}\hat{\vec{f}}_{\Phi}(\vec{x}),
\end{eqnarray*}
where  $\hat{\vec{f}}_{\Phi}(\vec{x})=(\vec{f}^{\vec{z}}_{\Phi}(\vec{x}))_{\vec{z}\neq \vec{0}}\in\real^{N_1}$
and $N_1=4^{n_1}-1$.  Similarly, for unital quantum channels $\Phi$, we have 
\begin{eqnarray}
\hat{\vec{f}}_{\Phi}(\vec{x})
=\hat{M}^{\Phi}
\hat{\vec{f}}_{I}(\vec{x}).
\end{eqnarray}
Therefore, we have

\begin{eqnarray*}
R_S(\mathcal{F}\circ \mathcal{C}^{n_0,n_1}_{\norm{\hat{\cdot}}\leq \mu})
&=&
\mathbb{E}_{\vec{\epsilon}}
\frac{1}{m}\sup_{\Phi\in \mathcal{C}^{n_0,n_1}_{\norm{\hat{\cdot}}\leq \mu}}
\left| \sum^m_{i=1}\epsilon_i\hat{\vec{\alpha}}\hat{\vec{f}}_{\Phi}(\vec{x}_i)\right|\\
&\leq&
\mathbb{E}_{\vec{\epsilon}}
\frac{1}{m}\sup_{\Phi}
\frac{\mu}{\norm{\hat{M}^\Phi}_{p,q}}
\left| \sum^m_{i=1}\epsilon_i\hat{\vec{\alpha}}\hat{\vec{f}}_{\Phi}(\vec{x}_i)\right|\\
&=&\mu\mathbb{E}_{\vec{\epsilon}}
\frac{1}{m}\sup_{\Phi}
\frac{1}{\norm{\hat{M}^\Phi}_{p,q}}
\left| \sum^m_{i=1}\epsilon_i\hat{\vec{\alpha}}\hat{\vec{f}}_{\Phi}(\vec{x}_i)\right|\\
&\leq&
\mu\norm{\hat{\vec{\alpha}}}_p\mathbb{E}_{\vec{\epsilon}}
\frac{1}{m}\sup_{\Phi}
\frac{1}{\norm{\hat{M}^\Phi}_{p,q}}
\norm{ \sum^m_{i=1}\epsilon_i\hat{\vec{f}}_{\Phi}(\vec{x}_i)}_{p^*}\\
&\leq&
\mu \norm{\hat{\vec{\alpha}}}_p\mathbb{E}_{\vec{\epsilon}}
\frac{1}{m}\sup_{\Phi}
\frac{1}{\norm{\hat{M}^\Phi}_{p,q}}
\norm{ \sum^m_{i=1}\epsilon_i\hat{M}^{\Phi}\hat{\vec{f}}_{I}(\vec{x}_i)}_{p^*}\\
&\leq&\mu \norm{\hat{\vec{\alpha}}}_p N^{\max\set{\frac{1}{p^*},\frac{1}{q}}}_1\mathbb{E}_{\vec{\epsilon}}
\frac{1}{m}
\norm{\sum^m_{i=1}\epsilon_i\hat{\vec{f}}_{I}(\vec{x}_i)}_{p^*},
\end{eqnarray*}
where the third inequality come from the Lemma \ref{lem:normpRaN1}.
Using Lemma \ref{lem:depn1}, we get the results of the theorem.

\end{proof}

\section{Deep quantum circuits}\label{apen:deep}
Consider a  depth-$l$ quantum circuit, where  
each layer of the quantum circuit is a treated as a quantum channel. 
We denote the depth-$l$ quantum circuit as $\vec{C}_l$ as follows
\begin{eqnarray}
\vec{C}_l=(\Phi_l, \Phi_{l-1}, \cdots, \Phi_1)
\end{eqnarray}
where the $i$-th layer $\Phi_i: \mathcal{L}((\complex^2)^{\ot n_{i-1}})\to \mathcal{L}((\complex^2)^{\ot n_i})$ (See Figure \ref{fig1}).

Let us define $\mathcal{C}^{l,\vec{n}}$ with the width vector $\vec{n}=(n_l,...,n_0)$ to be the set of 
all depth-$l$ quantum circuits $\vec{C}_{l}=(\Phi_l, \Phi_{l-1}, \cdots, \Phi_1)$, where the $i$-th layer $\Phi_i: \mathcal L((\complex^2)^{\ot n_{i-1}})\to \mathcal L((\complex^2)^{\ot n_i})$.
In this section, we introduce three resource measures to quantify the amount of magic in
quantum circuits by making use of the $(p,q)$ group norm.

 \begin{figure}[!ht]
  \center{\includegraphics[width=10cm]  {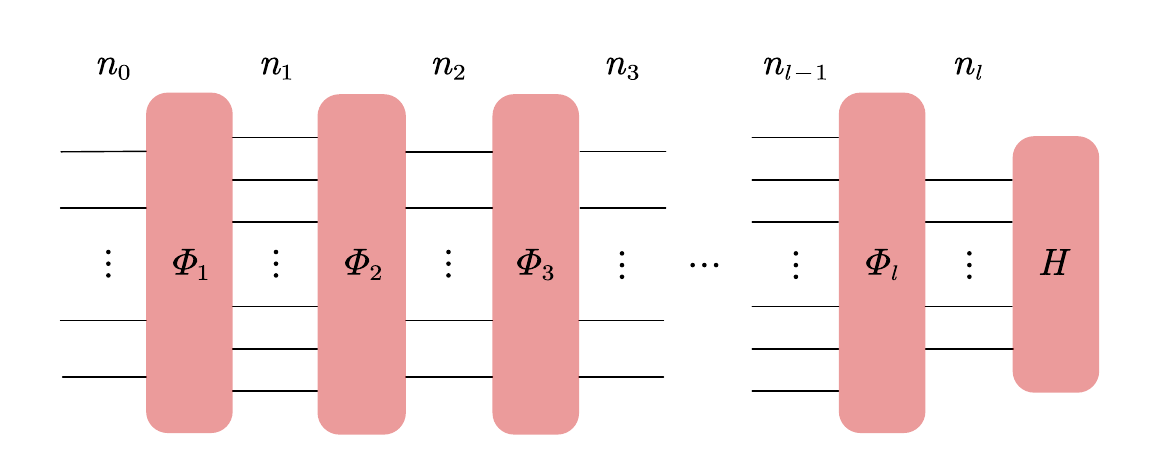}}     
  \caption{A diagram of a depth-$l$ quantum circuit}
  \label{fig1}
 \end{figure}

\subsection{Multiplication \texorpdfstring{$(p,q)$}{(p,q)} depth-norm}
In this subsection, let us define the multiplication $(p,q)$ depth-norm for  depth-$l$ quantum circuits
$\vec{C}_l=(\Phi_l,\Phi_{l-1},\cdots,\Phi_1)$
as follows
\begin{eqnarray}
\mu_{p,q}(\vec{C}_l)
=\prod^{l}_{i=1}\norm{M^{\Phi_i}}_{p,q}.
\end{eqnarray}

\begin{prop}\label{prop:prop_mult_p}
The multiplication $(p,q)$ depth-norm satisfies the following properties:

(1) Given a depth-$l$ quantum circuit $\vec{C}_l$ and a depth-$m$ quantum circuit $\vec{C}_m$, we have 
\begin{eqnarray}
\mu_{p,q}(\vec{C}_l\circ \vec{C}_m)
=\mu_{p,q}(\vec{C}_l)\mu_{p,q}(\vec{C}_m),
\end{eqnarray}
where $\vec{C}_l\circ \vec{C}_m:=(\vec{C}_l, \vec{C}_m)$.

(2) Given two depth-$l$ quantum circuits $C_l$ and $C'_l$, we have
\begin{eqnarray}
\mu_{p,q}(\vec{C}_l\ot \vec{C}'_l)
=\mu_{p,q}(\vec{C}_l)\mu_{p,q}(\vec{C}'_l).
\end{eqnarray}
where $\vec{C}_l\ot \vec{C}'_l:=(\Phi_l\ot \Phi'_l,\ldots,\Phi_1\ot\Phi'_1)$
for $\vec{C}_l=(\Phi_l,\ldots,\Phi_1), \vec{C}'_l=(\Phi'_l,\ldots,\Phi'_1)$.
\end{prop}
\begin{proof}
These two properties follow directly from the 
definition of $\mu_{p,q}$.
\end{proof}
Note that for the depth-$l$ quantum circuit $\vec{C}_l$, where each layer contains only unitary gates, i.e., $\vec{C}_l=(U_l, U_{l-1}, \cdots, U_1)$, $\mu_{p,q}$  can be viewed as a resource measure of magic.
\begin{lem}
Given a depth-$l$ quantum circuit $\vec{C}_l=(U_l, U_{l-1}, \cdots, U_1)$, we have 

(1) (Faithfulness) For $0<p<2$, $q>0$, it holds that 
$\mu_{p,q}(\vec{C}_l)\geq 1$, and 
 $\mu_{p,q}(\vec{C}_l)=1$ iff 
 $\vec{C}_l$ is a Clifford circuit, i.e., each $U_i$ is a Clifford unitary.

(1') (Faithfulness) For $p>2$, $0<q<\infty$, it holds that
$\mu_{p,q}(\vec{C}_l)\leq 1$, and  $\mu_{p,q}(\vec{C}_l)=1$ iff 
 $\vec{C}_l$ is a Clifford circuit.
 
 
 (2) (Invariance under Clifford circuit) For $p>0$, $q>0$, we have $\mu_{p,q}(\vec{C}_1\circ \vec{C}_l\circ \vec{C}_2)=\mu_{p,q}(\vec{C}_l)$ if $\vec{C}_1$, $\vec{C}_2$ are Clifford circuits.
\end{lem}
\begin{proof}
This lemma follows directly from Lemma \ref{lem:clif_pq} and  Proposition \ref{prop:prop_mult_p}.
\end{proof}

Next, let us denote the set of depth-$l$ quantum circuits $\vec{C}_l$ with bounded depth-norm $\mu_{p,q}$ as
$\mathcal{C}^{l,\vec{n}}_{\mu_{p,q}\leq \mu}$, that is,
\begin{eqnarray}
\mathcal{C}^{l,\vec{n}}_{\mu_{p,q}\leq \mu}:=
\set{\vec{C}_l\in \mathcal{C}^{l,\vec{n}}: \mu_{p,q}(\vec{C}_l)\leq \mu}.
\end{eqnarray}

\begin{lem}\label{lem:dep_norm1}
Given the set of depth-$l$ quantum circuits $\mathcal{C}^{l,\vec{n}}$
and the set of depth-$l$ quantum circuits $\mathcal{C}^{l,\vec{n'}}$, where $\vec{n}=(\vec{n}',n_l)$ and 
$\vec{n}'=(\vec{n}'',n_{l-1})$, then $ \forall \vec{\epsilon}\in \set{\pm1}^m$, we have 
\begin{eqnarray}
\sup_{\vec{C}_l\in \mathcal{C}^{l,\vec{n}}}
\frac{1}{\mu_{p,q}(\vec{C}_l)}
\norm{ \sum^m_{i=1}\epsilon_i\vec{f}_{C_l}(\vec{x}_i)}_{p^*}
\leq 4^{n_l\max\set{\frac{1}{p^*},\frac{1}{q}}}_{l}
\sup_{\vec{C}_{l-1}\in \mathcal{C}^{l,\vec{n}}}
\frac{1}{\mu_{p,q}(\vec{C}_{l-1})}
\norm{ \sum^m_{i=1}\epsilon_i\vec{f}_{C_{l-1}}(\vec{x}_i)}_{p^*}.
\end{eqnarray}
Thus, 
\begin{eqnarray}
\sup_{\vec{C}_l\in \mathcal{C}^{l,\vec{n}}}
\frac{1}{\mu_{p,q}(\vec{C}_l)}
\norm{ \sum^m_{i=1}\epsilon_if_{C_l}(\vec{x}_i)}_{p^*}
\leq
\prod^{l}_{i=1}4^{n_i\max\set{\frac{1}{p^*},\frac{1}{q}}}
\norm{\sum^m_{i=1}\epsilon_if_{I}(\vec{x}_i)}_{p^*}.
\end{eqnarray}

\end{lem}

\begin{proof}
The lemma follows from
\begin{eqnarray*}
\sup_{\vec{C}_l\in \mathcal{C}^{l,\vec{n}}}
\frac{1}{\mu_{p,q}(\vec{C}_l)}
\norm{ \sum^m_{i=1}\epsilon_i\vec{f}_{C_l}(\vec{x}_i)}_{p^*}
&=&\sup_{\vec{C}_{l-1}\in \mathcal{C}^{l-1,\vec{n}'}}
\frac{1}{\mu_{p,q}(\vec{C}_{l-1})\norm{M^{\Phi_l}}_{p,q}}
\norm{ \sum^m_{i=1}\epsilon_iM^{\Phi_l}\vec{f}_{C_{l-1}}(\vec{x}_i)}_{p^*}\\
&\leq& 4^{n_l\max\set{\frac{1}{p^*},\frac{1}{q}}}
\sup_{\vec{C}_{l-1}\in \mathcal{C}^{l-1,\vec{n}'}}
\frac{1}{\mu_{p,q}(\vec{C}_{l-1})}
\norm{\sum^m_{i=1}\epsilon_i\vec{f}_{C_{l-1}}(\vec{x}_i)}_{p^*},
\end{eqnarray*}
where the inequality follows from Lemma \ref{lem:normpRaN1}.

\end{proof}
\begin{thm}\label{thm:main1}
Given the set of depth-$l$ quantum circuits with bounded depth-norm $\mu_{p,q}$,
 the Rademacher complexity 
 on $m$ samples $S=\set{\vec{x_1},\ldots,\vec{x}_m}$ satisfies the following bounds 
 
 (1) For $1\leq p\leq 2$, we have
\begin{eqnarray}
R_S(\mathcal{F}\circ \mathcal{C}^{l,\vec{n}}_{\mu_{p,q}\leq \mu})
\leq \mu 4^{(\sum^{l}_{i=1}n_i)\max\set{\frac{1}{p^*},\frac{1}{q}}}
\frac{\sqrt{\min\set{p^*, 8n_0}}}{\sqrt{m}}\norm{\vec{\alpha}}_p\max_i\norm{\vec{f}_I(\vec{x}_i)}_{p^*}.
\end{eqnarray}

(2) For $2<p<\infty$, we have
\begin{eqnarray}
R_S(\mathcal{F}\circ \mathcal{C}^{l,\vec{n}}_{\mu_{p,q}\leq \mu})
\leq \mu 4^{(\sum^{l}_{i=1}n_i)\max\set{\frac{1}{p^*},\frac{1}{q}}}
\frac{\sqrt{p^*}}{m^{1/p}}\norm{\vec{\alpha}}_p\max_i\norm{\vec{f}_I(\vec{x}_i)}_{p^*}.
\end{eqnarray}

\end{thm}

\begin{proof}
These bounds follow from
\begin{eqnarray*}
R_S(\mathcal{F}\circ \mathcal{C}^{l,\vec{n}}_{\mu_{p,q}\leq \mu})
&=&
\mathbb{E}_{\vec{\epsilon}}
\frac{1}{m}\sup_{\vec{C}_l\in \mathcal{C}^{l,\vec{n}}_{\mu_{p,q}\leq \mu}}
\left| \sum^m_{i=1}\epsilon_i\vec{\alpha}\vec{f}_{C_l}(\vec{x}_i)\right|\\
&\leq&
\mathbb{E}_{\vec{\epsilon}}
\frac{1}{m}\sup_{\vec{C}_l\in \mathcal{C}^{l,\vec{n}}}
\frac{\mu}{\mu_{p,q}(\vec{C}_l)}
\left| \sum^m_{i=1}\epsilon_i\vec{\alpha}\vec{f}_{C_l}(\vec{x}_i)\right|\\
&=&\mu\mathbb{E}_{\vec{\epsilon}}
\frac{1}{m}\sup_{\vec{C}_l\in \mathcal{C}^{l,\vec{n}}}
\frac{1}{\mu_{p,q}(\vec{C}_l)}
\left| \sum^m_{i=1}\epsilon_i\vec{\alpha}\vec{f}_{C_l}(\vec{x}_i)\right|\\
&\leq&
\mu\norm{\vec{\alpha}}_p\mathbb{E}_{\vec{\epsilon}}
\frac{1}{m}\sup_{C_l\in \mathcal{C}^{l,\vec{n}}}
\frac{1}{\mu_{p,q}(\vec{C}_l)}
\norm{\sum^m_{i=1}\epsilon_i\vec{f}_{C_l}(\vec{x}_i)}_{p^*}
\\
&\leq&\mu \norm{\vec{\alpha}}_p\prod^{l}_{i=1}4^{n_i\max\set{\frac{1}{p^*},\frac{1}{q}}}\mathbb{E}_{\vec{\epsilon}}
\frac{1}{m}
\norm{\sum^m_{i=1}\epsilon_i\vec{f}_{I}(\vec{x}_i)}_{p^*},
\end{eqnarray*}
where the the second inequality comes from Lemma \ref{lem:dep_norm1}. 
Using Lemma \ref{lem:depn1}, we obtain the results of the theorem.
\end{proof}

To get rid of the exponential dependence on the width of the quantum neural network, 
we need to take $p^*\geq \sum^{l-1}_{i=1}n_i$, $q\geq \sum^{l-1}_{i=1}n_i$. For example,
we could
take $p=1, q=\infty$.

\begin{prop}
Given the set of depth-$l$ quantum circuits with bounded  $\mu_{1,\infty}$ norm,
 the Rademacher complexity 
 on $m$ samples $S=\set{\vec{x_1},\ldots,\vec{x}_m}$ satisfies the following bounds:

\begin{eqnarray}
R_S(\mathcal{F}\circ \mathcal{C}^{l,\vec{n}}_{\mu_{1,\infty}\leq \mu})
\leq \mu
\frac{\sqrt{ 8n_0}}{\sqrt{m}}\norm{\vec{\alpha}}_1\max_i\norm{\vec{f}_I(\vec{x}_i)}_{\infty}.
\end{eqnarray}

\end{prop}

We denote the set of  depth-$l$ variational quantum circuits with parameters $\vec{\theta}$ and fixed structure 
$\mathcal{A}$ by $\mathcal{C}^{l,n}_{\mathcal{A},\vec{\theta}}$. By `fixed structure', we mean that the position of each parametrized gate is fixed.  
Then, 
 for the $i$-th layer of the
variational quantum circuit, denoted as $\Phi_i(\vec{\theta})$,  let us define
$\mu_i:=\sup_{\vec{\theta}_i}\norm{M^{\Phi_i(\vec{\theta})}}_{1,\infty}$.
Therefore, for any such depth-$l$ variational quantum circuit with a fixed structure, we have 
$\mu_{p,q}\leq \prod_i \mu_i$. It follows that the Rademacher complexity of the class of depth-$l$ variational quantum circuits with fixed structure
is bounded by 
\begin{eqnarray}
R_S(\mathcal{F}\circ \mathcal{C}^{l,n}_{\mathcal{A},\vec{\theta}})
\leq \prod_i \mu_i\frac{\sqrt{ 8n_0}}{\sqrt{m}}\norm{\vec{\alpha}}_1\max_i\norm{\vec{f}_I(\vec{x}_i)}_{\infty}.
\end{eqnarray}

\subsection{Summation \texorpdfstring{$(p,q)$}{(p,q)} depth-norm}
In this subsection, let us define the summation
$(p,q)$ depth-norm for depth-$l$ quantum circuits $\vec{C}_l=(\Phi_l,...,\Phi_1)$ as follows:
\begin{eqnarray}
\nu^{(r)}_{p,q}(\vec{C}_l)
=\left( \frac{1}{l}\sum^l_{i=1}\norm{M^{\Phi_i}}^r_{p,q}\right)^{\frac{1}{r}},
\end{eqnarray}
for any $r> 0$. For example, if we take $r=1$, then
\begin{eqnarray}
\nu_{p,q}(\vec{C}_l)
=\frac{1}{l}\sum^{l}_{i=1}\norm{M^{\Phi_i}}_{p,q},
\end{eqnarray}
which is the average value of the amount of resources in each layer of the quantum circuit. 

\begin{prop}\label{prop:prop_sum_p}
The  summation $(p,q)$ depth-norm satisfy the following properties:

(1) Given a depth-$l$ quantum circuit $\vec{C}_l$ and a depth-$m$ quantum circuit $\vec{C}_m$, we have 
\begin{eqnarray}
(l+m)(\nu^{(r)}_{p,q}(\vec{C}_l\circ \vec{C}_m))^r
=l(\nu^{(r)}_{p,q}(\vec{C}_l))^r+m(\nu^{(r)}_{p,q}(\vec{C}_m))^r.
\end{eqnarray}

(2) Given two depth-$l$ quantum circuits $C_l$ and $C'_l$, we have
\begin{eqnarray}
\nu^{(r)}_{p,q}(\vec{C}_l\ot \vec{C}'_l)
\leq \nu^{(s)}_{p,q}(\vec{C}_l)\nu^{(t)}_{p,q}(\vec{C}'_l),
\end{eqnarray}
where $r,s,t>0$ and $\frac{1}{s}+\frac{1}{t}=\frac{1}{r}$.
\end{prop}
\begin{proof}
(1) follows directly from the definition of $\nu^{(r)}_{p,q}$, and 
(2) follows directly from H\"older's inequality. 
\end{proof}

Similarly,  for the depth-$l$ quantum circuit $\vec{C}_l$, where each layer only contains 
unitary gates, i.e., $\vec{C}_l=(U_l, U_{l-1}, \cdots, U_1)$, $\nu^r_{p,q}$  can be viewed as a resource measure of magic. 
\begin{lem}
Given a depth-$l$ quantum circuits $\vec{C}_l=(U_l, U_{l-1}, \cdots, U_1)$, we have 

(1) (Faithfulness) For $0<p<2$, $q>0$, $r>0$: 
$\nu^{(r)}_{p,q}(\vec{C}_l)\geq 1$, and 
 $\nu^{(r)}_{p,q}(\vec{C}_l)=1$ iff 
 $\vec{C}_l$ is a Clifford circuit.

(1') (Faithfulness) For $p>2$, $0<q<\infty$, $r>0$:
$\nu^{(r)}_{p,q}(\vec{C}_l)\leq 1$, and  $\nu^{(r)}_{p,q}(\vec{C}_l)=1$ iff 
 $\vec{C}_l$ is a Clifford circuit.

 (2) (Nonincreasing under Clifford circuits) For $0<p<2, q>0,r>0$, we have $\nu^{(r)}_{p,q}(\vec{C}_1\circ \vec{C}_l\circ \vec{C}_2)\leq \nu^{(r)}_{p,q}(\vec{C}_l)$ if $\vec{C}_1$, $\vec{C}_2$ are Clifford circuit.
 
 (2')  (Nondecreasing under Clifford circuits)
  For $p>2, 0<q<\infty,r>0$, we have $\nu^{(r)}_{p,q}(\vec{C}_1\circ \vec{C}_l\circ \vec{C}_2)\geq \nu^{(r)}_{p,q}(\vec{C}_l)$ if $\vec{C}_1$, $\vec{C}_2$ are Clifford circuit.

\end{lem}
\begin{proof}
This lemma comes directly from Proposition  \ref{lem:clif_pq} and \ref{prop:prop_sum_p}.

\end{proof}
Note that for $p>2, 0<q<\infty$, we can define the resource measure as $1-\nu^{(r)}_{p,q}$, in which case the resource measure also satisfies the properties of 
faithfulness and nonincreasing-ness under Clifford circuits.

Let us define the set of depth-$l$ quantum circuits with bounded  $\nu^{(r)}_{p,q}$ norm by $\mathcal{C}^{l,\vec{n}}_{\nu^{(r)}_{p,q}\leq \nu}$.
It is easy to see the following relationship 
$\nu^{(r)}_{p,q}$ and $\mu_{p,q}$
\begin{eqnarray}
\nu^{(r)}_{p,q}(\vec{C}_l)
\geq \mu_{p,q}(\vec{C}_l)^{1/l},
\end{eqnarray}
which follows directly from the Arithmetic Mean-Geometric Mean inequality. 
Hence we have 
\begin{eqnarray}
\mathcal{C}^{l,\vec{n}}_{\nu^{(r)}_{p,q}\leq \nu}\subseteq \mathcal{C}^{l,\vec{n}}_{\mu_{p,q}\leq \nu^l}.
\end{eqnarray}

This allows us to obtain the following result on the Rademacher complexity of quantum circuits with bounded  $\nu^r_{p,q}$ norm directly from Theorem \ref{thm:main1}.

\begin{thm}[Restatement of Theorem \ref{thm:Con_main1}]
Given the set of depth-$l$ quantum circuits with bounded $\nu^{(r)}_{p,q}$,
the Rademacher complexity 
 on $m$ independent samples $S=\set{\vec{x_1},...,\vec{x}_m}$ satisfies the following bounds 
 
 (1) For $1\leq p\leq 2$, we have
\begin{eqnarray}
R_S(\mathcal{F}\circ \mathcal{C}^{l,\vec{n}}_{\nu^{(r)}_{p,q}\leq \nu})
\leq \nu^l 4^{(\sum^{l}_{i=1}n_i)\max\set{\frac{1}{p^*},\frac{1}{q}}}
\frac{\sqrt{\min\set{p^*, 8n_0}}}{\sqrt{m}}\norm{\vec{\alpha}}_p\max_i\norm{\vec{f}_I(\vec{x}_i)}_{p^*}.
\end{eqnarray}

(2) For $2<p<\infty$, we have
\begin{eqnarray}
R_S(\mathcal{F}\circ \mathcal{C}^{l,\vec{n}}_{\nu^{(r)}_{p,q}\leq \nu})
\leq \nu^l 4^{(\sum^{l}_{i=1}n_i)\max\set{\frac{1}{p^*},\frac{1}{q}}}
\frac{\sqrt{p^*}}{m^{1/p}}\norm{\vec{\alpha}}_p\max_i\norm{\vec{f}_I(\vec{x}_i)}_{p^*}.
\end{eqnarray}

\end{thm}
To get rid of the exponential dependence on the width of quantum neural networks, 
we need to take $p^*\geq \sum^{l-1}_{i=1}n_i$ and $q\geq \sum^{l-1}_{i=1}n_i$. For example,
we could 
take $p=1, q=\infty$.

\begin{prop}
Given the set of depth-$l$ quantum circuits with bounded $\nu^{(r)}_{1,\infty}$, the Rademacher complexity 
 on $m$ independent samples $S=\set{\vec{x_1},\ldots,\vec{x}_m}$ satisfies the following bounds 
 
\begin{eqnarray}
R_S(\mathcal{F}\circ \mathcal{C}^{l,\vec{n}}_{\nu^{(r)}_{1,\infty}\leq \nu})
\leq \nu^l 
\frac{\sqrt{ 8n_0}}{\sqrt{m}}\norm{\vec{\alpha}}_p\max_i\norm{\vec{f}_I(\vec{x}_i)}_{\infty}.
\end{eqnarray}

\end{prop}

\subsection{\texorpdfstring{$(p,q)$}{(p,q)} path norm}

Let us define the $(p,q)$ path-norm for the depth-$l$ quantum circuits 
$\vec{C}_l=(\Phi_l, \Phi_{l-1},...,\Phi_1)$. First, for a fixed output $P_{\vec{z}}$,
where $\vec{z}\in\set{0,1,2,3}^{n_l}$,  let us define
\begin{eqnarray}
\gamma^{(\vec{z})}_{p}(\vec{C}_l)
=\left(
\sum_{\substack{v_0\to v_1\to..\to v_{out},\\
v_{out}=\vec{z}}}
|M^{\Phi_l}_{\vec{z}v_{l-1}}M^{\Phi_{l-1}}_{v_{l-1}v_{l-2}}\cdots M^{\Phi_{1}}_{v_{1}v_{0}} |^p
\right)^{1/p}.
\end{eqnarray}
Hence, we can define the $(p,q)$ path-norm for the depth-$l$ quantum circuits   as follows,
\begin{eqnarray}
\gamma_{p,q}(\vec{C}_l)
=\left(\frac{1}{4^{n_l}}\sum_{\vec{z}}\gamma^{(\vec{z})}_p(\vec{C}_l)^q\right)^{1/q}.
\end{eqnarray}

\begin{prop}\label{prop:prop_path_p}
The multiplication $(p,q)$ path-norm satisfies the following properties:

(1) Given a depth-$l$ quantum circuit $\vec{C}_l$ and a depth-$m$ quantum circuit $\vec{C}_m$, we have 
\begin{eqnarray}
\gamma_{p,q}(\vec{C}_l\circ \vec{C}_m)
\leq \gamma_{p,q}(\vec{C}_l)\gamma_{p,\infty}(\vec{C}_m).
\end{eqnarray}

(2) Given two depth-$l$ quantum circuits $\vec{C}_l$ and $\vec{C}'_l$, we have
\begin{eqnarray}
\gamma_{p,q}(\vec{C}_l\ot \vec{C}'_l)
=\gamma_{p,q}(\vec{C}_l)\gamma_{p,q}(\vec{C}'_l).
\end{eqnarray}
\end{prop}
\begin{proof}\hfill

(1) holds because 
\begin{eqnarray*}
\gamma^{(\vec{z})}_p(\vec{C}_l\circ \vec{C}_m)
&=&\left(\sum_{v_0\to\ldots\to v_{m}
\to v_{m+1}\to...\to v_{l+m}=\vec{z}}
|M_{\vec{z}\vec{v}_{l+m-1}}|^p
\ldots|M_{\vec{v}_{m+1}\vec{v}_m}|^p
\ldots|M_{\vec{v}_1\vec{v}_0}|^p\right)^{1/p}\\
&=&\left(\sum_{v_{m}
\to v_{m+1}\to\ldots\to v_{l+m}=\vec{z}}
|M_{\vec{z}\vec{v}_{l+m-1}}|^p
\ldots|M_{\vec{v}_{m+1}\vec{v}_m}|^p (\gamma^{\vec{v}_m}_{p}(\vec{C}_m))^p\right)^{1/p}\\
&\leq& \left(\sum_{v_{m}
\to v_{m+1}\to \ldots \to v_{l+m}=\vec{z}}
|M_{\vec{z}\vec{v}_{l+m-1}}|^p
\ldots|M_{\vec{v}_{m+1}\vec{v}_m}|^p \right)^{1/p} \max_{\vec{v}_m}\gamma^{(\vec{v}_m)}_{p}(\vec{C}_m)\\
&=&\gamma^{(\vec{z})}_{p}(\vec{C}_l)\gamma_{p,\infty}(\vec{C}_m).
\end{eqnarray*}
Therefore, we have $\gamma_{p,q}(\vec{C}_l\circ \vec{C}_m)
\leq \gamma_{p,q}(\vec{C}_l)\gamma_{p,\infty}(\vec{C}_m)$.

And (2) holds because
\begin{eqnarray*}
\gamma^{(\vec{z}_1\vec{z}_2)}_{p}(\vec{C}_l\ot \vec{C}'_l)
&=&\left(
\sum_{v_0\to v_1\to v_2\to\ldots\to v_{out},
v_{out}=\vec{z}_1\vec{z}_2}
|M^{\Phi_l\ot \Phi'_l}_{\vec{z}v_{l-1}}M^{\Phi_{l-1}\ot\Phi'_{l-1}}_{v_{l-1}v_{l-2}}\cdots M^{\Phi_{1}\ot\Phi'_1}_{v_{1}v_{0}} |^p
\right)^{1/p}\\
&=&\gamma^{(\vec{z}_1)}_{p}(\vec{C}_l)\gamma^{(\vec{z}_2)}_{p}(\vec{C}'_l).
\end{eqnarray*}
where the second equality comes from the fact that $M^{\Phi\ot\Phi'}=M^{\Phi}\ot M^{\Phi'}$.

\end{proof}

\begin{prop}\label{prop:gamma_C}
For any depth-$l$ quantum circuit $\vec{C}_l$, we have the following relationship:
For any $0<p\leq 1$, $q>0$, we have
\begin{eqnarray}
\gamma_{p,q}(\vec{C}_l)
\geq \norm{M^{C_l}}_{p,q}.
\end{eqnarray}
\end{prop}

\begin{proof}
To prove this result, we only need to prove that 
for any $\vec{z}$, we have 
\begin{eqnarray*}
\gamma^{(\vec{z})}_{p}(\vec{C}_l)
\leq \norm{M^{C_l}_{\vec{z}}}_p.
\end{eqnarray*}
This is because 
\begin{eqnarray*}
\norm{M^{C_l}_{\vec{z}}}_p
&=&
\left(
\sum_{\vec{v_0}}|\sum_{\vec{v}_1,...,\vec{v}_{l-1}}M^{\Phi_l}_{\vec{x}\vec{v}_{l-1}}M^{\Phi_{l-1}}_{\vec{v}_{l-1}\vec{v}_{l-2}}\cdots M^{\Phi_{1}}_{\vec{v}_{1}\vec{v}_{0}} |^p
\right)^{1/p}\\
&\leq& \left(
\sum_{\vec{v_0}}\sum_{\vec{v}_1,...,\vec{v}_{l-1}}|M^{\Phi_l}_{\vec{x}\vec{v}_{l-1}}M^{\Phi_{l-1}}_{\vec{v}_{l-1}\vec{v}_{l-2}}\cdots M^{\Phi_{1}}_{\vec{v}_{1}\vec{v}_{0}} |^p
\right)^{1/p}\\
&=&\gamma^{(\vec{z})}_p(\vec{C}_l).
\end{eqnarray*}
\end{proof}

For a depth-$l$ quantum circuit $\vec{C}_l$, where each layer contains only
unitary gates, i.e., $\vec{C}_l=(U_l, U_{l-1},\cdots,  U_1)$, the path norm $\gamma_{p,q}$  can be viewed as a resource measure of magic.

\begin{lem}\label{lem:crit_path_duplicate}
Given a depth-$l$ quantum circuit $\vec{C}_l=(U_l, U_{l-1},\cdots, U_1)$, we have 

(1) (Faithfulness) For $0<p\leq 1$, $q>0$: $\gamma_{p,q}(\vec{C}_l)\geq 1$, 
 $\gamma_{p,q}(\vec{C}_l)= 1$ iff $\vec{C}_l$ is a Clifford circuit.

(2)  (Invariance under Clifford circuit)
$
\gamma_{p,q}(\vec{C}_1\circ \vec{C}_l\circ \vec{C}_2)
=\gamma_{p,q}(\vec{C}_l)
$ if $\vec{C}_1$ and $\vec{C}_2$ are Clifford circuits.
\end{lem}

\begin{proof}
$\gamma_{p,q}(\vec{C}_l)\geq 1$
comes from the facts that 
$\gamma_{p,q}(\vec{C}_l)
\geq \norm{M^{C_l}}_{p,q}$ and $\norm{M^{C_l}}_{p,q}\geq 1$ (by Lemma \ref{lem:clif_pq}).

Finally, the invariance under Clifford circuits has been proved in 
Proposition \ref{prop:prop_path_p}.

\end{proof}

Let us define the normalized  representation matrix  of the quantum channel in the depth-$l$ quantum circuit $C_l$ as follows
\begin{eqnarray}
m^{\Phi_{k+1}}_{\vec{z}\vec{x},p}
=\frac{M^{\Phi_{k+1}}_{\vec{z}\vec{x}}\gamma^{(\vec{x})}_{p}(\vec{C}_k)}{\gamma^{(\vec{z})}_{p}(\vec{C}_{k+1})}.
\end{eqnarray}

It is easy to see that for any row vector 
$m^{\Phi_{k+1}}_{\vec{z},p}$, we have 
\begin{eqnarray}
\norm{m^{\Phi_{k+1}}_{\vec{z},p}}_p
=\left(\sum_{\vec{x}}|m^{\Phi_{k+1}}_{\vec{z}\vec{x},p}|^p\right)^{1/p}
=1, \forall \vec{z}.
\end{eqnarray}
Besides, it is easy to verify that
\begin{eqnarray}
\gamma^{(\vec{z})}_{p}(\vec{C}_l)m^{\Phi_{l}}_{\vec{z}}m^{\Phi_{l-1}}....m^{\Phi_{1}}
=M^{\Phi_l}_{\vec{z}}
M^{\Phi_{l-1}}
\cdots 
M^{\Phi_1}.
\end{eqnarray}
Therefore
\begin{eqnarray}
f_{C_l}(\vec{x})
=\vec{\alpha}\vec{f}_{C_l}(\vec{x})
=\vec{\alpha}D(\gamma(\vec{C}_l))
\vec{\tilde{f}}_{C_l}(\vec{x}),
\end{eqnarray}
where
$D(\gamma(C_l))=\diag(\gamma^{(\vec{z})}_p)_{\vec{z}}$, 
$\vec{\tilde{f}}_{C_l}(\vec{x})=m^{\Phi_{l}}\vec{\tilde{f}}_{C_{l}}(\vec{x})$, 
$\vec{\tilde{f}}_{C_1}(\vec{x})=m^{\Phi_{1}}\vec{f}_{I}(\vec{x})$.
It is easy to see that 
\begin{eqnarray}
\norm{D(C_l)m^{\Phi_l}}_{p,q}
=\left(\frac{1}{N}\sum_{\vec{z}}\gamma^{(\vec{z})}_{p}(C_l)^q
\right)^{\frac{1}{q}}.
\end{eqnarray}

\begin{lem}\label{lem:pathn}
For any $p^*>0$, the following statement holds for any $k$-depth quantum circuits
$\vec{C}_k=(\Phi_k, \Phi_{k-1},..., \Phi_1)$, 
\begin{eqnarray}
\mathbb{E}_{\vec{\epsilon}}
\frac{1}{m}\sup_{\vec{C}_{k}\in \mathcal{C}^{k, \vec{n}}}
\norm{\sum^m_{i=1}\epsilon_i\vec{\tilde{f}}_{C_{k}}(\vec{x}_i)}_{p^*}
\leq 4^{n_k\frac{1}{p^*}}\mathbb{E}_{\vec{\epsilon}}
\frac{1}{m}
\sup_{\vec{C}_{k-1}\in \mathcal{C}^{k-1, \vec{n}}}\norm{\sum^m_{i=1}\epsilon_i\vec{\tilde{f}}_{C_{k-1}}(\vec{x}_i)}_{p^*}.
\end{eqnarray}
\end{lem}

\begin{proof}
This lemma comes from the fact that
\begin{eqnarray*}
\mathbb{E}_{\vec{\epsilon}}
\frac{1}{m}\sup_{\vec{C}_{k}\in \mathcal{C}^{k, \vec{n}}}
\norm{\sum^m_{i=1}\epsilon_i\vec{\tilde{f}}_{\vec{C}_{k}}(\vec{x}_i)}_{p^*}
&\leq&4^{n_k\frac{1}{p^*}}\mathbb{E}_{\vec{\epsilon}}
\frac{1}{m}\sup_{\vec{C}_{k}\in \mathcal{C}^{k, \vec{n}}}
\norm{\sum^m_{i=1}\epsilon_i\vec{\tilde{f}}_{C_{k}}(\vec{x}_i)}_{\infty}\\
&=&4^{n_k\frac{1}{p^*}}\mathbb{E}_{\vec{\epsilon}}
\frac{1}{m}\sup_{\vec{C}_{k-1}\in \mathcal{C}^{k-1, \vec{n}}}
\sup_{\vec{z}}\left|\sum^m_{i=1}\epsilon_im^{\Phi_{k}}_{\vec{z}}\vec{\tilde{f}}_{C_{k-1}}(\vec{x}_i)\right|\\
&\leq& 4^{n_k\frac{1}{p^*}}\mathbb{E}_{\vec{\epsilon}}
\frac{1}{m}
\sup_{\vec{C}_{k-1}\in \mathcal{C}^{k-1, \vec{n}}}\norm{\sum^m_{i=1}\epsilon_i\vec{\tilde{f}}_{C_{k-1}}(\vec{x}_i)}_{p^*}.
\end{eqnarray*}

\end{proof}

\begin{thm}\label{thm:main2duplicate}
Given the set of depth-$l$ quantum circuits with bounded path norm $\gamma_{p,q}$,
the Rademacher complexity 
 on $m$ independent samples $S=\set{\vec{x_1},\ldots,\vec{x}_m}$ satisfies the following bounds 
 
 (1) For $1\leq p\leq 2$, we have
\begin{eqnarray}
R_S(\mathcal{F}\circ \mathcal{C}^{l,\vec{n}}_{\gamma_{p,q}(\vec{C}_l)\leq \gamma})
\leq \gamma 4^{n_l\max\set{\frac{1}{p^*},\frac{1}{q}}}
\prod^{l-1}_{i=1}4^{n_i\frac{1}{p^*}}
\frac{\sqrt{\min\set{p^*, 8n_0}}}{\sqrt{m}}\norm{\vec{\alpha}}_p\max_i\norm{\vec{f}_I(\vec{x}_i)}_{p^*}.
\end{eqnarray}

(2) For $2<p<\infty$, we have
\begin{eqnarray}
R_S(\mathcal{F}\circ \mathcal{C}^{l,\vec{n}}_{\gamma_{p,q}(\vec{C}_l)\leq \gamma})
\leq \gamma 4^{n_l\max\set{\frac{1}{p^*},\frac{1}{q}}}
\prod^{l-1}_{i=1}4^{n_i\frac{1}{p^*}}
\frac{\sqrt{p^*}}{m^{1/p}}\norm{\vec{\alpha}}_p\max_i\norm{\vec{f}_I(\vec{x}_i)}_{p^*}.
\end{eqnarray}

\end{thm}

\begin{proof}

First we compute the following
\begin{eqnarray*}
R_S(\mathcal{F}\circ \mathcal{C}^{l,\vec{n}}_{\gamma_{p,q}(\vec{C}_l)\leq \gamma})
&=&
\mathbb{E}_{\vec{\epsilon}}
\frac{1}{m}\sup_{\vec{C}_l\in \mathcal{C}^{l,\vec{n}}_{\gamma_{p,q}(\vec{C}_l)\leq \gamma}}
\left| \sum^m_{i=1}\epsilon_i\vec{\alpha}\vec{f}_{C_l}(\vec{x}_i)\right|\\
&=&
\mathbb{E}_{\vec{\epsilon}}
\frac{1}{m}\sup_{\vec{C}_l\in \mathcal{C}^{l,\vec{n}}}
\frac{\gamma}{\gamma_{p,q}(\vec{C}_l)}
\left| \sum^m_{i=1}\epsilon_i\vec{\alpha}D(C_l)\vec{\tilde{f}}_{C_l}(\vec{x}_i)\right|
\\
&\leq&\norm{\vec{\alpha}}_p
\mathbb{E}_{\vec{\epsilon}}
\frac{1}{m}\sup_{\vec{C}_l\in \mathcal{C}^{l,\vec{n}}}\frac{\gamma}{\gamma_{p,q}(\vec{C}_l)}
\norm{\sum^m_{i=1}\epsilon_iD(C_l)\vec{\tilde{f}}_{C_l}(\vec{x}_i)}_{p^*}
\\
&=&\norm{\vec{\alpha}}_p
\mathbb{E}_{\vec{\epsilon}}
\frac{1}{m}\sup_{\vec{C}_l\in \mathcal{C}^{l,\vec{n}}}\frac{\gamma}{\gamma_{p,q}(\vec{C}_l)}
\norm{\sum^m_{i=1}\epsilon_iD(C_l)m^{\Phi_l}\vec{\tilde{f}}_{C_{l-1}}(\vec{x}_i)}_{p^*}
\\
&\leq&\norm{\vec{\alpha}}_p
\mathbb{E}_{\vec{\epsilon}}
\frac{1}{m}\sup_{\vec{C}_l\in \mathcal{C}^{l,\vec{n}}}\frac{\gamma}{\gamma_{p,q}(\vec{C}_l)}
4^{n_l\max\set{\frac{1}{p^*},\frac{1}{q}}}\norm{D(C_l)m^{\Phi_l}}_{p,q}
\norm{\sum^m_{i=1}\epsilon_i\vec{\tilde{f}}_{C_{l-1}}(\vec{x}_i)}_{p^*}\\
&=&\gamma\norm{\vec{\alpha}}_p 4^{n_l\max\set{\frac{1}{p^*},\frac{1}{q}}}
\mathbb{E}_{\vec{\epsilon}}
\frac{1}{m}\sup_{\vec{C}_l\in \mathcal{C}^{l,\vec{n}}}\norm{\sum^m_{i=1}\epsilon_i\vec{\tilde{f}}_{C_{l-1}}(\vec{x}_i)}_{p^*}
\\
&\leq&
\gamma \norm{\vec{\alpha}}_p4^{n_l\max\set{\frac{1}{p^*},\frac{1}{q}}}_l
\prod^{l-1}_{i=1}4^{n_i\frac{1}{p^*}}
\mathbb{E}_{\vec{\epsilon}}
\frac{1}{m}\norm{\sum^m_{i=1}\epsilon_i\vec{\tilde{f}}_{I}(\vec{x}_i)}_{p^*},
\end{eqnarray*}
where the last inequality comes from Lemma \ref{lem:pathn}. The theorem follows from this and Lemma \ref{lem:depn1}.

\end{proof}

\section{Deep unital quantum circuits}\label{apen:deep_unital}

In this section, let us consider  depth-$l$ unital quantum circuits, where
each layer of the quantum circuit is a quantum channel. Furthermore, we shall assume that 
the observable $H$ is traceless.
Unlike the previous section, we consider the $(p,q) $ norm of the modified representation matrix, where
$\mathbb{I}$, and only $\mathbb{I}$, is mapped to $\mathbb{I}$.

 \subsection{Modified multiplication \texorpdfstring{$(p,q)$}{(p,q)} depth-norm}
Let us define the modified multiplication  $(p,q)$ depth-norm for  depth-$l$ unital quantum circuits
$\vec{C}_l=(\Phi_l, \Phi_{l-1}, \cdots, \Phi_1)$
as follows
\begin{eqnarray}
\hat{\mu}_{p,q}(\vec{C}_l)
=\prod^{l}_{i=1}\norm{\hat{M}^{\Phi_i}}_{p,q}.
\end{eqnarray}

\begin{lem}
Given a depth-$l$ quantum circuit $\vec{C}_l=(U_l, U_{l-1}, \cdots, U_1)$, we have 

(1) (Faithfulness) For $0<p<2$, $q>0$: 
$\hat{\mu}_{p,q}(\vec{C}_l)\geq 1$, and 
 $\hat{\mu}_{p,q}(\vec{C}_l)=1$ iff 
 $\vec{C}_l$ is a Clifford circuit.

(1') (Faithfulness) For $p>2$, $q>0$:
$\hat{\mu}_{p,q}(\vec{C}_l)\leq 1$, and  $\hat{\mu}_{p,q}(\vec{C}_l)=1$ iff 
 $\vec{C}_l$ is a Clifford circuit.

 (2) (Invariance under Clifford circuit) For $p>0, q>0$, we have $\hat{\mu}_{p,q}(\vec{C}_1\circ \vec{C}_l\circ \vec{C}_2)=\hat{\mu}_{p,q}(\vec{C}_l)$ if $\vec{C}_1$, $\vec{C}_2$ are Clifford circuit.
\end{lem}
\begin{proof}
This lemma follows directly from Proposition  \ref{lem:hat_unitary} and \ref{prop:prop_mult_p}.
\end{proof}

\begin{lem}\label{lem:dep_normU1}
Given the set of depth-$l$ unital quantum circuit $\mathcal{C}^{l,\vec{n}}$
and the set of depth-$l$ unital quantum circuit $\mathcal{C}^{l,\vec{n'}}$, where $\vec{n}=(\vec{n}',n_l)$ and 
$\vec{n}'=(\vec{n}'',n_{l-1})$, then $ \forall \vec{\epsilon}\in \set{\pm1}^m$, we have 
\begin{eqnarray}
\sup_{\vec{C}_l\in \mathcal{C}^{l,\vec{n}}}
\frac{1}{\hat{\mu}_{p,q}(\vec{C}_l)}
\norm{ \sum^m_{i=1}\epsilon_i\hat{\vec{f}}_{C_l}(\vec{x}_i)}_{p^*}
\leq N^{\max\set{\frac{1}{p^*},\frac{1}{q}}}_{l}
\sup_{\vec{C}_{l-1}\in \mathcal{C}^{l,\vec{n}}}
\frac{1}{\hat{\mu}_{p,q}(\vec{C}_{l-1})}
\norm{ \sum^m_{i=1}\epsilon_i\hat{\vec{f}}_{C_{l-1}}(\vec{x}_i)}_{p^*},
\end{eqnarray}
where $N_{l}=4^{n_{l}}-1$.
Thus, 
\begin{eqnarray}
\sup_{\vec{C}_l\in \mathcal{C}^{l,\vec{n}}}
\frac{1}{\hat{\mu}_{p,q}(\vec{C}_l)}
\norm{ \sum^m_{i=1}\epsilon_i\hat{\vec{f}}_{C_l}(\vec{x}_i)}_{p^*}
\leq
\prod^{l}_{i=1}N^{\max\set{\frac{1}{p^*},\frac{1}{q}}}_i
\norm{\sum^m_{i=1}\epsilon_i\hat{\vec{f}}_{I}(\vec{x}_i)}_{p^*}.
\end{eqnarray}

\end{lem}
\begin{proof}
This is because

\begin{eqnarray*}
\sup_{\vec{C}_l\in \mathcal{C}^{l,\vec{n}}}
\frac{1}{\hat{\mu}_{p,q}(\vec{C}_l)}
\norm{ \sum^m_{i=1}\epsilon_i\hat{\vec{f}}_{C_l}(\vec{x}_i)}_{p^*}
&=&\sup_{\vec{C}_{l-1}\in \mathcal{C}^{l-1,\vec{n}'}}
\frac{1}{\hat{\mu}_{p,q}(\vec{C}_{l-1})\norm{\hat{M}^{\Phi_l}}_{p,q}}
\norm{ \sum^m_{i=1}\epsilon_i\hat{M}^{\Phi_l}\hat{\vec{f}}_{C_{l-1}}(\vec{x}_i)}_{p^*}\\
&\leq& N^{\max\set{\frac{1}{p^*},\frac{1}{q}}}_{l}
\sup_{\vec{C}_{l-1}\in \mathcal{C}^{l-1,\vec{n}'}}
\frac{1}{\hat{\mu}_{p,q}(\vec{C}_{l-1})}
\norm{\sum^m_{i=1}\epsilon_i\hat{\vec{f}}_{C_{l-1}}(\vec{x}_i)}_{p^*},
\end{eqnarray*}
where the inequality follows from Lemma \ref{lem:normpRaN1}.

\end{proof}

\begin{thm}\label{thm:mainU}
Given the set of depth-$l$ unital quantum circuits  with bounded depth-norm $\hat{\mu}_{p,q}$,
the Rademacher complexity 
 on $m$ samples $S=\set{\vec{x_1},\ldots,\vec{x}_m}$ satisfies the following bounds 
 
 (1) For $1\leq p\leq 2$, we have
\begin{eqnarray}
R_S(\mathcal{F}\circ \mathcal{C}^{l,\vec{n}}_{\hat{\mu}_{p,q}\leq \mu})
\leq \mu \prod^l_{i=1}N^{\max\set{\frac{1}{p^*},\frac{1}{q}}}_i
\frac{\sqrt{\min\set{p^*, 8n_0}}}{\sqrt{m}}\norm{\hat{\vec{\alpha}}}_p\max_i\norm{\hat{\vec{f}}_I(\vec{x}_i)}_{p^*}.
\end{eqnarray}

(2) For $2<p<\infty$, we have
\begin{eqnarray}
R_S(\mathcal{F}\circ \mathcal{C}^{l,\vec{n}}_{\hat{\mu}_{p,q}\leq \mu})
\leq \mu \prod^l_{i=1}N^{\max\set{\frac{1}{p^*},\frac{1}{q}}}_i
\frac{\sqrt{p^*}}{m^{1/p}}\norm{\hat{\vec{\alpha}}}_p\max_i\norm{\hat{\vec{f}}_I(\vec{x}_i)}_{p^*}.
\end{eqnarray}

where $N_i=4^{n_i}-1$.
\end{thm}

\begin{proof}
The theorem follows from
\begin{eqnarray*}
R_S(\mathcal{F}\circ \mathcal{C}^{l,\vec{n}}_{\hat{\mu}_{p,q}\leq \mu})
&=&
\mathbb{E}_{\vec{\epsilon}}
\frac{1}{m}\sup_{\vec{C}_l\in \mathcal{C}^{l,\vec{n}}_{\hat{\mu}_{p,q}\leq \mu}}
\left| \sum^m_{i=1}\epsilon_i\hat{\vec{\alpha}}\hat{\vec{f}}_{C_l}(\vec{x}_i)\right|\\
&\leq&
\mathbb{E}_{\vec{\epsilon}}
\frac{1}{m}\sup_{\vec{C}_l\in \mathcal{C}^{l,\vec{n}}}
\frac{\mu}{\hat{\mu}_{p,q}(\vec{C}_l)}
\left| \sum^m_{i=1}\epsilon_i\hat{\vec{\alpha}}\hat{\vec{f}}_{C_l}(\vec{x}_i)\right|\\
&=&\mu\mathbb{E}_{\vec{\epsilon}}
\frac{1}{m}\sup_{\vec{C}_l\in \mathcal{C}^{l,\vec{n}}}
\frac{1}{\hat{\mu}_{p,q}(\vec{C}_l)}
\left| \sum^m_{i=1}\epsilon_i\hat{\vec{\alpha}}\hat{\vec{f}}_{C_l}(\vec{x}_i)\right|\\
&\leq&
\mu\norm{\hat{\vec{\alpha}}}_p\mathbb{E}_{\vec{\epsilon}}
\frac{1}{m}\sup_{\vec{C}_l\in \mathcal{C}^{l,\vec{n}}}
\frac{1}{\hat{\mu}_{p,q}(\vec{C}_l)}
\norm{\sum^m_{i=1}\epsilon_i\hat{\vec{f}}_{C_l}(\vec{x}_i)}_{p^*}
\\
&\leq&\mu \norm{\hat{\vec{\alpha}}}_p\prod^{l}_iN^{\max\set{\frac{1}{p^*},\frac{1}{q}}}_i\mathbb{E}_{\vec{\epsilon}}
\frac{1}{m}
\norm{\sum^m_{i=1}\epsilon_i\hat{\vec{f}}_{I}(\vec{x}_i)}_{p^*},
\end{eqnarray*}
where the second inequality follows from Lemma \ref{lem:dep_normU1}. 
Using Lemma \ref{lem:depn1}, we get the results of theorem theorem.
\end{proof}

If we take $p=1, q=\infty$, 
then we have the following results directly from the previous results.
\begin{prop}
Given the set of depth-$l$ unital quantum circuits with bounded $\hat{\mu}_{1,\infty}$ norm,
 the Rademacher complexity 
 on $m$ samples $S=\set{\vec{x_1},\ldots,\vec{x}_m}$ satisfies the following bounds

\begin{eqnarray}
R_S(\mathcal{F}\circ \mathcal{C}^{l,\vec{n}}_{\hat{\mu}_{1,\infty}\leq \mu})
\leq \mu
\frac{\sqrt{ 8n_0}}{\sqrt{m}}\norm{\hat{\vec{\alpha}}}_p\max_i\norm{\hat{\vec{f}}_I(\vec{x}_i)}_{\infty}.
\end{eqnarray}
\end{prop}

Hence, consider a depth-$l$ variational unitary quantum circuit on $n$ qubits with parameters $\vec{\theta}$ and a fixed structure 
$\mathcal{A}$, where each layer $U_i(\vec{\theta})=\ot^{k_i}_{j=1}U^{(j)}_i(\vec{\theta}_j)$. Then 
based  on the properties of $(p,q)$ norm of modified matrix $\hat{M}$ of unitary channels,
\begin{eqnarray}
\mu_i:=\sup_{\vec{\theta}_i}\norm{\hat{M}^{U_i(\vec{\theta})}}_{1,\infty}
=\sup_{\vec{\theta}}\prod_j\norm{\hat{M}^{U^{(j)}_i(\vec{\theta_j})}}_{1,\infty}
=\prod_j\sup_{\vec{\theta}_j}\norm{\hat{M}^{U^{(j)}_i(\vec{\theta_j})}}_{1,\infty}
=\prod_j\mu^{(j)}_i.
\end{eqnarray}
Therefore, for any such depth-$l$ variational quantum circuits with fixed structure, we have 
$\hat{\mu}_{1,\infty}\leq \prod_i \prod_j\mu^{(j)}_i$ and each $\mu^{(j)}_i\geq 1$.

\begin{cor}
 The Rademacher complexity of the quantum circuits class of depth-$l$ variational quantum circuits with fixed structure
is bounded by 
\begin{eqnarray}
R_S(\mathcal{F}\circ \mathcal{C}^{l,n}_{\mathcal{A},\vec{\theta}})
\leq \prod_i \prod_j\mu^{(j)}_i\frac{\sqrt{ 8n}}{\sqrt{m}}\norm{\hat{\vec{\alpha}}}_p\max_i\norm{\vec{f}_I(\vec{x}_i)}_{\infty}.
\end{eqnarray}
\end{cor}

\subsection{Modified summation \texorpdfstring{$(p,q)$}{(p,q)} depth-norm}
Let us define the modified summation
$(p,q)$ depth-norm for depth-$l$ quantum circuits as follows
\begin{eqnarray}
\hat{\nu}^{r}_{p,q}(\vec{C}_l)
=\left( \frac{1}{l}\sum^l_{i=1}\norm{\hat{M}^{\Phi_i}}^r_{p,q}\right)^{\frac{1}{r}},
\end{eqnarray}
for any $r> 0$.

\begin{lem}
Given a depth-$l$ quantum circuit $\vec{C}_l=(U_l, U_{l-1}, \cdots, U_1)$, we have 

(1) (Faithfulness) For $0<p<2$, $q>0$, $r>0$: 
$\hat{\nu}^{(r)}_{p,q}(\vec{C}_l)\geq 1$, and 
 $\hat{\nu}^{(r)}_{p,q}(\vec{C}_l)=1$ iff 
 $\vec{C}_l$ is a Clifford circuit.

(1') (Faithfulness) For $p>2$, $q>0$, $r>0$:
$\hat{\nu}^{(r)}_{p,q}(\vec{C}_l)\leq 1$ and  $\hat{\nu}^{(r)}_{p,q}(\vec{C}_l)=1$ iff 
 $\vec{C}_l$ is a Clifford circuit.

 (2) (Nonincreasing under Clifford circuit) For $0<p<2, q>0,r>0$, we have $\hat{\nu}^{(r)}_{p,q}(\vec{C}_1\circ \vec{C}_l\circ \vec{C}_2)\leq \hat{\nu}^{(r)}_{p,q}(\vec{C}_l)$ if $\vec{C}_1$, $\vec{C}_2$ are Clifford circuit.
 
 (2')  (Nondecreasing under Clifford circuit)
  For $p>2, q>0,r>0$, we have $\hat{\nu}^{(r)}_{p,q}(\vec{C}_1\circ \vec{C}_l\circ \vec{C}_2)\geq \hat{\nu}^{(r)}_{p,q}(\vec{C}_l)$ if $\vec{C}_1$, $\vec{C}_2$ are Clifford circuit.

\end{lem}
\begin{proof}
This lemma comes directly from Lemma  \ref{lem:hat_unitary} and Proposition \ref{prop:prop_sum_p}.
\end{proof}

Based on the following relationship between 
$\hat{\nu}^{(r)}_{p,q}$ and $\hat{\mu}_{p,q}$
\begin{eqnarray}
\hat{\nu}^{(r)}_{p,q}(\vec{C}_l)
\geq \hat{\mu}_{p,q}(\vec{C}_l)^{1/l},
\end{eqnarray}
we have 
\begin{eqnarray}
\mathcal{C}^{l,\vec{n}}_{\hat{\nu}_{p,q}\leq \nu}\subseteq \mathcal{C}^{l,\vec{n}}_{\hat{\mu}_{p,q}\leq \nu^l}.
\end{eqnarray}

We obtain the following results directly from Theorem \ref{thm:mainU}.

\begin{thm}[Restatement of Theorem \ref{thm:Con_main2}]
Given the set of depth-$l$ unital quantum circuits with bounded $\hat{\nu}^{(r)}_{p,q}$ norm,
 the Rademacher complexity 
 on $m$ samples $S=\set{\vec{x_1},\ldots,\vec{x}_m}$ satisfies the following bounds 
 
 (1) For $1\leq p\leq 2$, we have
\begin{eqnarray}
R_S(\mathcal{F}\circ \mathcal{C}^{l,\vec{n}}_{\hat{\nu}^r_{p,q}\leq \nu})
\leq \nu^l \prod^l_{i=1}N^{\max\set{\frac{1}{p^*},\frac{1}{q}}}_i
\frac{\sqrt{\min\set{p^*, 8n_0}}}{\sqrt{m}}\norm{\hat{\vec{\alpha}}}_p\max_i\norm{\vec{f}_I(\vec{x}_i)}_{p^*}.
\end{eqnarray}

(2) For $2<p<\infty$, we have
\begin{eqnarray}
R_S(\mathcal{F}\circ \mathcal{C}^{l,\vec{n}}_{\hat{\nu}^{(r)}_{p,q}\leq \nu})
\leq \nu^l \prod^l_{i=1}N^{\max\set{\frac{1}{p^*},\frac{1}{q}}}_i
\frac{\sqrt{p^*}}{m^{1/p}}\norm{\hat{\vec{\alpha}}}_p\max_i\norm{\vec{f}_I(\vec{x}_i)}_{p^*}.
\end{eqnarray}

\end{thm}

If we take $p=1, q=\infty$, 
then we get the following results directly from the previous results.
\begin{prop}
Given the set of depth-$l$ unital quantum circuits with bounded  $\hat{\nu}^r_{1,\infty}$ norm, the Rademacher complexity 
 on $m$ samples $S=\set{\vec{x_1},\ldots,\vec{x}_m}$ satisfies the following bounds

\begin{eqnarray}
R_S(\mathcal{F}\circ \mathcal{C}^{l,\vec{n}}_{\hat{\nu}^r_{1,\infty}\leq \nu})
\leq \nu^l 
\frac{\sqrt{ 8n_0}}{\sqrt{m}}\norm{\hat{\vec{\alpha}}}_p\max_i\norm{\hat{\vec{f}}_I(\vec{x}_i)}_{\infty}.
\end{eqnarray}
\end{prop}

\subsection{Modified \texorpdfstring{$(p,q)$}{(p,q)} path-norm}
Let us define the modified $(p,q)$ path-norm for depth-$l$ circuits by 
\begin{eqnarray}
\hat{\gamma}_{p,q}(\vec{C}_l)
=\left(\frac{1}{N_l}\sum_{\vec{x}\neq \vec{0}}\gamma^{(\vec{x})}_p(\vec{C}_l)^q\right)^{1/q},
\end{eqnarray}
where $N_l=4^{n_l}-1$.

Let us define the normalized  representation matrix  of quantum channels in depth-$l$ quantum circuits $\vec{C}_l$ as follows
\begin{eqnarray}
\hat{m}^{\Phi_{k+1}}_{\vec{z}\vec{x},p}
=\frac{\hat{M}^{\Phi_{k+1}}_{\vec{z}\vec{x}}\gamma^{(\vec{x})}_{p}(\vec{C}_k)}{\gamma^{(\vec{z})}_{p}(\vec{C}_{k+1})}.
\end{eqnarray}

It is easy to see that for any row vector 
$\hat{m}^{\Phi_{k+1}}_{\vec{z},p}$, we have 
\begin{eqnarray}
\norm{\hat{m}^{\Phi_{k+1}}_{\vec{z},p}}_p
=\left(\sum_{\vec{x}}|\hat{m}^{\Phi_{k+1}}_{\vec{z}\vec{x},p}|^p\right)^{1/p}
=1, \forall \vec{z}.
\end{eqnarray}
Besides, it is easy to verify that
\begin{eqnarray}
\gamma^{(\vec{z})}_{p}(\vec{C}_l)\hat{m}^{\Phi_{l}}_{\vec{z}}\hat{m}^{\Phi_{l-1}}....\hat{m}^{\Phi_{1}}
=\hat{M}^{\Phi_l}_{\vec{z}}
\hat{M}^{\Phi_{l-1}}
\cdots 
\hat{M}^{\Phi_1},
\end{eqnarray}
Therefore,
\begin{eqnarray}
f_{C_l}(\vec{x})
=\hat{\vec{\alpha}}\hat{\vec{f}}_{C_l}(\vec{x})
=\hat{\vec{\alpha}}\hat{D}(\gamma(\vec{C}_l))
\hat{\vec{\tilde{f}}}_{C_l}(\vec{x}),
\end{eqnarray}
where $\hat{D}(\gamma(C_l))=\diag(\gamma^{\vec{z}}_p)_{\vec{z}\neq \vec{0}}$,
$\hat{\vec{\tilde{f}}}=m^{\Phi_{l}}\hat{\vec{\tilde{f}}}_{C_{l}}(\vec{x})$, and 
$\hat{\vec{\tilde{f}}}_{C_1}(\vec{x})=m^{\Phi_{1}}\hat{\vec{f}}_{I}(\vec{x})$.
It is easy to see that 
\begin{eqnarray}
\norm{\hat{D}(C_l)\hat{m}^{\Phi_l}}_{p,q}
=\left(\frac{1}{N_l}\sum_{\vec{z}\neq\vec{0}}\gamma^{(\vec{z})}_{p}(\vec{C}_l)^q
\right)^{\frac{1}{q}}.
\end{eqnarray}

Similarly to $\gamma_{p,q}$, $\hat{\gamma}_{p,q}$ satisfies the following property.
\begin{prop}
For any depth-$l$ unital quantum circuit $\vec{C}_l$, we have the following relationship:
For any $0<p\leq 1$, $q>0$, we have
\begin{eqnarray}
\hat{\gamma}_{p,q}(\vec{C}_l)
\geq \norm{\hat{M}^{C_l}}_{p,q}.
\end{eqnarray}
\end{prop}
\begin{proof}
The proof is similar to that of Proposition \ref{prop:gamma_C}.
\end{proof}

\begin{lem}\label{lem:crit_path}
Given a depth-$l$ quantum circuit $\vec{C}_l=(U_l, U_{l-1}, \cdots, U_1)$, we have 

(1) (Faithfulness) For $0<p\leq 1$, $q>0$, $\hat{\gamma}_{p,q}(\vec{C}_l)\geq 1$, 
 $\gamma_{p,q}(\vec{C}_l)= 1$ iff $\vec{C}_l$ is Clifford.

(2)  (Invariance under Clifford circuit)
$
\hat{\gamma}_{p,q}(\vec{C}_1\circ \vec{C}_l\circ \vec{C}_2)
=\hat{\gamma}_{p,q}(\vec{C}_l)
$ if $\vec{C}_1$ and $\vec{C}_2$ are Clifford circuits.
\end{lem}
\begin{proof}
The proof is similar to that of Lemma \ref{lem:crit_path_duplicate}.
\end{proof}

\begin{thm}\label{thm:main2}
Given the set of depth-$l$ unital quantum circuits with bounded path norm $\gamma_{p,q}$, the Rademacher complexity 
 on $m$ samples $S=\set{\vec{x_1},\ldots,\vec{x}_m}$ satisfies the following bounds 
 
 (1) For $1\leq p\leq 2$, we have
\begin{eqnarray}
R_S(\mathcal{F}\circ \mathcal{C}^{l,\vec{n}}_{\hat{\gamma}_{p,q}(C_l)\leq \gamma})
\leq \gamma N^{\max\set{\frac{1}{p^*},\frac{1}{q}}}_l
\prod^{l-1}_{i=1}N^{\frac{1}{p^*}}_i
\frac{\sqrt{\min\set{p^*, 8n_0}}}{\sqrt{m}}\norm{\hat{\vec{\alpha}}}_p\max_i\norm{\vec{f}_I(\vec{x}_i)}_{p^*}.
\end{eqnarray}

(2) For $2<p<\infty$, we have
\begin{eqnarray}
R_S(\mathcal{F}\circ \mathcal{C}^{l,\vec{n}}_{\hat{\gamma}_{p,q}(C_l)\leq \gamma})
\leq \gamma N^{\max\set{\frac{1}{p^*},\frac{1}{q}}}_l
\prod^{l-1}_{i=1}N^{\frac{1}{p^*}}_i
\frac{\sqrt{p^*}}{m^{1/p}}\norm{\hat{\vec{\alpha}}}_p\max_i\norm{\vec{f}_I(\vec{x}_i)}_{p^*},
\end{eqnarray}
where $N_i=4^{n_i}-1$.
\end{thm}

\begin{proof}

The statement in the theorem holds because
\begin{eqnarray*}
R_S(\mathcal{F}\circ \mathcal{C}^{l,\vec{n}}_{\hat{\gamma}_{p,q}(C_l)\leq \gamma})
&=&
\mathbb{E}_{\vec{\epsilon}}
\frac{1}{m}\sup_{C_l\in \mathcal{C}^{l,\vec{n}}_{\hat{\gamma}_{p,q}(C_l)\leq \gamma}}
\left| \sum^m_{i=1}\epsilon_i\hat{\vec{\alpha}}\vec{f}_{C_l}(\vec{x}_i)\right|\\
&=&
\mathbb{E}_{\vec{\epsilon}}
\frac{1}{m}\sup_{C_l\in \mathcal{C}^{l,\vec{n}}}
\frac{\gamma}{\hat{\gamma}_{p,q}(C_l)}
\left| \sum^m_{i=1}\epsilon_i\hat{\vec{\alpha}}\hat{D}(C_l)\vec{\tilde{\hat{f}}}_{C_l}(\vec{x}_i)\right|
\\
&\leq&\norm{\hat{\vec{\alpha}}}_p
\mathbb{E}_{\vec{\epsilon}}
\frac{1}{m}\sup_{C_l\in \mathcal{C}^{l,\vec{n}}}\frac{\gamma}{\hat{\gamma}_{p,q}(C_l)}
\norm{\sum^m_{i=1}\epsilon_iD(C_l)\vec{\tilde{\hat{f}}}_{C_l}(\vec{x}_i)}_{p^*}
\\
&=&\norm{\hat{\vec{\alpha}}}_p
\mathbb{E}_{\vec{\epsilon}}
\frac{1}{m}\sup_{C_l\in \mathcal{C}^{l,\vec{n}}}\frac{\gamma}{\hat{\gamma}_{p,q}(C_l)}
\norm{\sum^m_{i=1}\epsilon_i\hat{D}(C_l)\hat{m}^{\Phi_l}\vec{\tilde{\hat{f}}}_{C_{l-1}}(\vec{x}_i)}_{p^*}
\\
&\leq&\norm{\hat{\vec{\alpha}}}_p
\mathbb{E}_{\vec{\epsilon}}
\frac{1}{m}\sup_{C_l\in \mathcal{C}^{l,\vec{n}}}\frac{\gamma}{\hat{\gamma}_{p,q}(C_l)}
(N-1)^{\max\set{\frac{1}{p^*},\frac{1}{q}}}\norm{\hat{D}(C_l)\hat{m}^{\Phi_l}}_{p,q}
\norm{\sum^m_{i=1}\epsilon_i\vec{\tilde{\hat{f}}}_{C_{l-1}}(\vec{x}_i)}_{p^*}\\
&=&\gamma \norm{\hat{\vec{\alpha}}}_p N^{\max\set{\frac{1}{p^*},\frac{1}{q}}}
\mathbb{E}_{\vec{\epsilon}}
\frac{1}{m}\sup_{C_l\in \mathcal{C}^{l,\vec{n}}}\norm{\sum^m_{i=1}\epsilon_i\vec{\tilde{\hat{f}}}_{C_{l-1}}(\vec{x}_i)}_{p^*}
\\
&\leq&
\gamma \norm{\hat{\vec{\alpha}}}_p N^{\max\set{\frac{1}{p^*},\frac{1}{q}}}
\prod^{l-1}_{i=1}N^{\frac{1}{p^*}}_i
\mathbb{E}_{\vec{\epsilon}}
\frac{1}{m}\norm{\sum^m_{i=1}\epsilon_i\vec{\tilde{f}}_{I}(\vec{x}_i)}_{p^*} ,
\end{eqnarray*}
where the last inequality follows from Lemma \ref{lem:pathn}.
Using Lemma \ref{lem:depn1} completes the proof of the theorem.

\end{proof}

\end{document}